\documentclass[a4paper,twocolumn,10pt,accepted=2024-06-25]{quantumarticle}
\pdfoutput=1
\usepackage[utf8]{inputenc}
\usepackage[english]{babel}
\usepackage[T1]{fontenc}
\usepackage{amsmath}
\usepackage{xcolor}
\definecolor{myrefcolor}{rgb}{0.067,0.5,0.5}
\usepackage[breaklinks,pdftex,colorlinks=true,citecolor=myrefcolor]{hyperref}
\usepackage{bbm,braket,microtype,mathrsfs,amsmath,amssymb,color,amsthm,mathtools,fullpage,graphicx,enumitem,ae,aecompl,bm,thm-restate,dsfont}
\usepackage{tikz}
\usepackage{lipsum}
\usepackage[sort&compress,numbers]{natbib}
\usepackage{algorithm}
\usepackage{algorithmic}

\newcommand{\be}{\begin{equation}}
\newcommand{\ee}{\end{equation}}
\DeclareMathOperator{\tr}{\rm{Tr}}

\newcommand{\ignore}[1]{}
\newcommand{\st}[1]{\ket{#1}\!\!\bra{#1}}

\newcommand{\Var}{{\rm Var}}
\newcommand{\Cov}{{\rm Cov}}

\newcommand{\de}[0]{{\operatorname{d}}}

\newcommand{\expf}[1]{\mathrm{exp}\left ( {#1}\right )}

\newcommand{\aver}[1]{ \left\langle  {#1}  \right\rangle }

\newcommand{\gde}{\operatorname{GDE}}
\newcommand{\goe}{\operatorname{GOE}}
\newcommand{\gue}{\operatorname{GUE}}
\newcommand{\poi}{\operatorname{P}}

\newcommand{\sym}{\text{sym}}

\newcommand{\poly}{\operatorname{poly}}
\newcommand{\supp}{\operatorname{supp}}

\newcommand{\cnot}{\operatorname{CNOT}}

\def\norm#1{\left\| #1\right\|}
\def\CC{{\rm\kern.24em \vrule width.04em height1.46ex depth-.07ex
   \kern-.29em C}}
\def\P{{\rm I\kern-.25em P}}
\def\RR{{\rm
        \vrule width.04em height1.58ex depth-.0ex
        \kern-.04em R}}
\def\bbbc{{\mathchoice {\setbox0=\hbox{$\displaystyle\rm C$}\hbox{\hbox
to0pt{\kern0.4\wd0\vrule height0.9\ht0\hss}\box0}}
{\setbox0=\hbox{$\textstyle\rm C$}\hbox{\hbox
to0pt{\kern0.4\wd0\vrule height0.9\ht0\hss}\box0}}
{\setbox0=\hbox{$\scriptstyle\rm C$}\hbox{\hbox
to0pt{\kern0.4\wd0\vrule height0.9\ht0\hss}\box0}}
{\setbox0=\hbox{$\scriptscriptstyle\rm C$}\hbox{\hbox
to0pt{\kern0.4\wd0\vrule height0.9\ht0\hss}\box0}}}}
\def\bbbz{{\mathchoice {\hbox{$\sf\textstyle Z\kern-0.4em Z$}}
{\hbox{$\sf\textstyle Z\kern-0.4em Z$}}
{\hbox{$\sf\scriptstyle Z\kern-0.3em Z$}}
{\hbox{$\sf\scriptscriptstyle Z\kern-0.2em Z$}}}}

\newlength{\fighskip} \fighskip=2pt
\newlength{\figvskip} \figvskip=1pt

\renewcommand{\vec}[1]{\boldsymbol{#1}} 
\DeclareMathOperator*{\argmin}{arg\,min}
\newcommand{\Tr}{{\rm Tr}}
\newcommand{\dya}[1]{\ket{#1}\!\bra{#1}}

\newtheorem{theorem}{Theorem}
\newtheorem{prop}{Proposition}
\newtheorem{corollary}{Corollary}
\newtheorem{lemma}{Lemma}

\newtheorem{definition}{Definition}

\begin{document}

\title{On the practical usefulness of the Hardware Efficient Ansatz}

\author{Lorenzo Leone}
\affiliation{Theoretical Division (T-4), Los Alamos National Laboratory, Los Alamos, New Mexico 87545, USA}
\affiliation{Center for Nonlinear Studies, Los Alamos National Laboratory, Los Alamos, New Mexico 87545, USA}
\affiliation{Physics Department, University of Massachusetts Boston, Boston, Massachusetts 02125, USA}
\email{Lorenzo.Leone001@umb.edu}
\orcid{0000-0002-0334-7419}
\author{Salvatore F.E. Oliviero}
\affiliation{Theoretical Division (T-4), Los Alamos National Laboratory, Los Alamos, New Mexico 87545, USA}
\email{s.oliviero001@umb.edu}
\affiliation{Center for Nonlinear Studies, Los Alamos National Laboratory, Los Alamos, New Mexico 87545, USA}
\affiliation{Physics Department, University of Massachusetts Boston, Boston, Massachusetts 02125, USA}
\orcid{0000-0003-3569-085X}

\author{Lukasz Cincio}
\affiliation{Theoretical Division (T-4), Los Alamos National Laboratory, Los Alamos, New Mexico 87545, USA}
\orcid{0000-0002-6758-4376}
\author{M. Cerezo}
\affiliation{Information Sciences, Los Alamos National Laboratory, Los Alamos, NM 87545, USA}
\orcid{0000-0002-2757-3170}

\begin{abstract}
Variational Quantum Algorithms (VQAs) and Quantum Machine Learning (QML) models train a parametrized quantum circuit to solve a given learning task. The success of these algorithms greatly hinges on appropriately choosing an ansatz for the quantum circuit. Perhaps one of the most famous ansatzes is the one-dimensional layered Hardware Efficient Ansatz (HEA), which seeks to minimize the effect of hardware noise by using native gates and connectives. The use of this HEA has generated a certain ambivalence arising from the fact that while it suffers from barren plateaus at long depths, it can also avoid them at shallow ones. In this work, we attempt to determine whether one should, or should not, use a HEA. We rigorously identify scenarios where shallow HEAs should likely be avoided (e.g.,   VQA or QML tasks with data satisfying a volume law of entanglement). More importantly, we identify a Goldilocks scenario where shallow HEAs could achieve a quantum speedup: QML tasks with data satisfying an area law of entanglement. We provide examples for such scenario (such as Gaussian diagonal ensemble random Hamiltonian discrimination), and we show that in these cases a shallow HEA is always trainable and that there exists an anti-concentration of loss function values.  Our work highlights the crucial role that input states play in the trainability of a parametrized quantum circuit, a phenomenon that is verified in our numerics. 
\end{abstract}
\section{Introduction}

The advent of Noisy Intermediate-Scale Quantum (NISQ)~\cite{preskill2018quantum} computers has generated a tremendous amount of excitement. Despite the presence of hardware noise and their limited qubit count,  near-term quantum computers are already capable of outperforming the world's largest super-computers on certain contrived mathematical tasks~\cite{google2019supremacy,wu2021strong,madsen2022quantum}. This has started a veritable rat race to solve real-life tasks of interest in NISQ hardware.

One of the most promising strategies to make practical use of near-term quantum computers is to train parametrized hybrid quantum-classical models. Here, a quantum device is used to estimate a classically hard-to-compute quantity, while one also leverages classical optimizers to train the parameters in the model. When the algorithm is problem-driven, we usually refer to it as a Variational Quantum Algorithm (VQA)~\cite{cerezo2020variationalreview,bharti2021noisy}. VQAs can be used for a wide range of tasks such as finding the ground state of molecular Hamiltonians~\cite{peruzzo2014variational,arute2020hartree}, solving combinatorial optimization tasks~\cite{farhi2014quantum,arute2020quantum} and solving linear systems of equations~\cite{bravo2020variational,huang2019near,xu2019variational}, among others. On the other hand, when the algorithm is data-driven, we refer to it as a Quantum Machine Learning (QML) model~\cite{biamonte2017quantum,schuld2021machine}. QML can be used in supervised~\cite{havlivcek2019supervised,schatzki2021entangled}, unsupervised~\cite{otterbach2017unsupervised} and reinforced~\cite{jerbi2021parametrized} learning problems, where the data processed in the quantum device can either be classical data embedded in quantum states~\cite{havlivcek2019supervised,perez2020data}, or quantum data obtained from some physical   process~\cite{cong2019quantum,caro2021generalization,huang2021provably}.

Both VQAs and QML models train parametrized quantum circuits $U(\vec{\theta})$ to solve their respective tasks. One of, if not \textit{the}, most important aspect in determining the success of these near-term algorithms is the choice of ansatz for the parametrized quantum circuit~\cite{cerezo2022challenges}. By ansatz, we mean the specifications for the arrangement and type of quantum gates in $U(\vec{\theta})$, and how these depend on the set of trainable parameters $\vec{\theta}$. Recently, the field of ansatz design has seen a Cambrian explosion where researchers have proposed a plethora of ansatzes for VQAs and QML~\cite{cerezo2020variationalreview,bharti2021noisy}. These include variable structure ansatzes~\cite{zhu2020adaptive,tang2019qubit,zhang2021mutual,bilkis2021semi,rattew2019domain}, problem-inspired ansatzes~\cite{hadfield2019quantum,wiersema2020exploring,lee2021towards,verdon2019quantumgraph,bausch2020recurrent} and even the recently introduced field of geometric quantum machine learning where one embeds information about the data symmetries into $U(\vec{\theta})$~\cite{larocca2022group,meyer2022exploiting,skolik2022equivariant,sauvage2022building,ragone2022representation,nguyen2022atheory,schatzki2022theoretical}. Choosing an appropriate ansatz is crucial as it has been shown that ill-defined ansatzes can be untrainable~\cite{mcclean2018barren,cerezo2020cost,sharma2020trainability,thanasilp2021subtleties,holmes2021connecting,arrasmith2021equivalence,pesah2020absence,uvarov2020barren,marrero2020entanglement,patti2020entanglement}, and hence useless for large scale implementations.

\begin{figure*}[t]
    \centering
    \includegraphics[width=.99\linewidth]{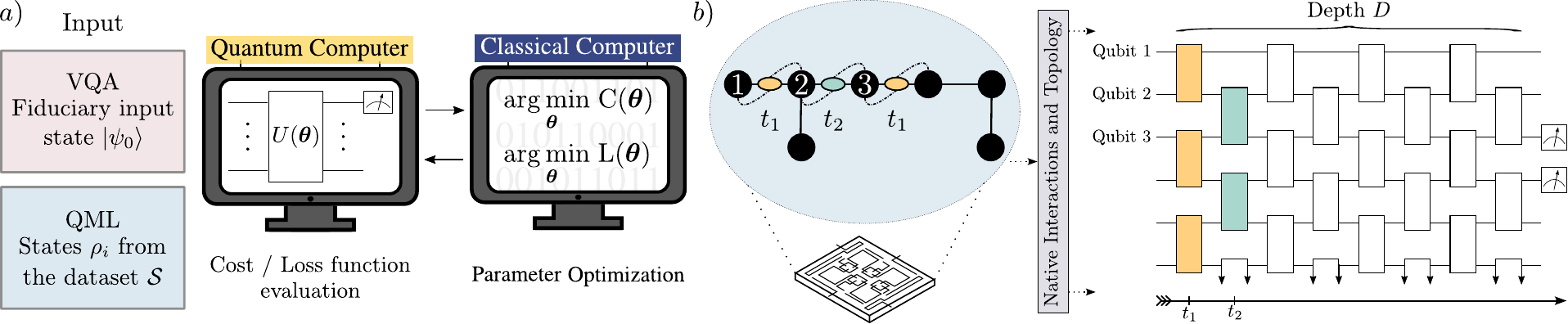}
    \caption{$(a)$ Both VQA and QML models train parametrized quantum circuits $U(\vec{\theta})$ to  minimize either a cost function $C(\vec{\theta})$ for VQAs, or a loss function $L(\vec{\theta})$ for QML models. While VQAs start from some fiduciary, easy to prepare state $\ket{\psi_0}$, in QML one uses states from  dataset $\ket{\psi_s}\in\mathcal{S}$ as input of the parametrized circuit $U(\vec{\theta})$. Both models exploit the power of classical optimizers for the minimization task. $(b)$ The architecture of a HEA seeks to minimize the effect of hardware noise by following the topology, and using the native gates, of the physical hardware. Specifically, we consider HEA as a one-dimensional alternating layered ansatz of two-qubit gates organized in a brick-like fashion. In the figure, we show how a first layer of gates is implemented at time $t_1$ while a second layer at time $t_2$. At the end of the computation, a local operator is measured.}
    \label{fig:HEA}
\end{figure*}

Perhaps the most famous, and simultaneously infamous, ansatz is the so-called Hardware Efficient Ansatz (HEA). As its name implies, the main objective of HEA is to mitigate the effect of hardware noise by using gates native to the specific device being used. The previous avoids the gate overhead which arises when compiling~\cite{khatri2019quantum} a non-native gate-set into a sequence of native gates. While the HEA was originally proposed within the framework of  VQAs, it is now also widely used in QML tasks. The strengths of the HEA are that it can be as depth-frugal as possible and that it is problem-agnostic, meaning that one can use it in any  scenario. However, its wide usability could also be its greater weakness, as it is believed that the HEA cannot have a good performance on all tasks~\cite{holmes2021connecting} (this is similar to the famous no-free-lunch theorem in classical machine learning~\cite{wolpert1997no}). Moreover,  it was shown that deep HEA circuits suffer from barren plateaus~\cite{mcclean2018barren} due to their high expressibility~\cite{holmes2021connecting}. Despite these difficulties, the HEA is not completely hopeless. In Ref.~\cite{cerezo2020cost}, the HEA saw a glimmer of hope as it was shown that shallow HEAs can be immune to barren plateaus, and thus have trainability guarantees.

From the previous, the HEA was left in a sort of gray-area of ansatzes, where its practical usefulness was unclear. On the one hand, there is a common practice in the field of using the HEA irrespective of the problem one is trying to solve. On the other hand, there is a significant push to move away from problem-agnostic HEA, and instead develop problem-specific ansatzes. However, the answer to questions such as ``\textit{Should we use (if at all) the HEA?}'' or ``\textit{What problems are shallow HEAs good for?}'' have not been rigorously tackled.

In this work, we attempt to determine what are the problems in VQAs and QML where HEAs should, or should not be used. As we will see, our results indicate that HEAs should likely be avoided in VQA tasks where the input state is a product state, as the ensuing algorithm can be efficiently simulated via classical methods. Similarly, we will rigorously prove that HEAs should not be used in QML tasks where the input data satisfies a volume law of entanglement. In these cases, we connect the entanglement in the input data to the phenomenon of cost concentration, and we show that high levels of entanglement lead to barren plateaus, and hence to  untrainability. Finally, we identify a scenario where shallow HEAs can be useful and potentially capable of achieving a quantum advantage: QML tasks where the input data satisfies an area law of entanglement. In these cases, we can guarantee that the optimization landscape will not exhibit barren plateaus. Taken together our results highlight the critical importance that the input data plays in the trainability of a model.

\section{Framework}

\subsection{Variational Quantum Algorithms and Quantum Machine Learning }

Throughout this work, we will consider two related, but conceptually different, hybrid quantum-classical models. The first, which we will denote as a Variational Quantum Algorithm (VQA) model, can be used to solve the following tasks
\begin{definition}[Variational Quantum Algorithms]\label{def:VQA}
Let $O$ be a Hermitian operator, whose ground state encodes the solution to a problem of interest. In a VQA task, the goal is to minimize a cost function $C(\vec{\theta})$, parametrized through a quantum circuit $U(\vec{\theta})$, to prepare the ground state of $O$ from a fiduciary state $\ket{\psi_0}$.
\end{definition}
In a VQA task one usually defines a cost function of the form
\begin{equation} \label{eq:cost}
    C(\vec{\theta})=\Tr[ U(\vec{\theta})\dya{\psi_0}  U^{\dagger}(\vec{\theta}) O]\,,
\end{equation}
and trains the parameters in $U(\vec{\theta})$ by solving the optimization task $   \argmin_{\vec{\theta}}C(\vec{\theta})$.

Then, while Quantum Machine Learning (QML) models can be used for a wide range of learning tasks, here we will focus on supervised problems
\begin{definition}[Quantum Machine Learning]\label{def:QML}
Let $\mathcal{S}=\{y_s,\ket{\psi_s}\}$ be a dataset of interest, where $\ket{\psi_s}$ are $n$-qubit states and $y_s$ associated real-valued labels. In a QML task, the goal is to train a model, by minimizing a loss function $L({\vec{\theta}})$  parametrized through a quantum neural network, i.e., a parametrized quantum circuit, $U(\vec{\theta})$, to predict labels that closely match those in the dataset.
\end{definition}
The exact form of $L({\vec{\theta}})$, and concomitantly the nature of what we want to ``learn'' from the dataset depends on the task at hand. For instance, in a binary QML classification task where $y_s$ are labels  one can minimize  an empirical loss function such as the mean-squared error  $L(\vec{\theta})=\frac{1}{|\mathcal{S}|}\sum_{s}(y_{s}-L_s(\vec{\theta}))^{2}$, or the hinge-loss where $L(\vec{\theta})=\frac{-1}{|\mathcal{S}|}\sum_{s}y_{s}L_s(\vec{\theta})$. Here,  
\be \label{eq:loss}
L_s(\vec{\theta})=\Tr[ U(\vec{\theta})\dya{\psi_s}  U^{\dagger}(\vec{\theta}) O_{s}]\,,
\ee
with $O_{s}$ being label-dependent Hermitian operator. The parameters in the quantum neural network  $U(\vec{\theta})$ are trained by solving the optimization task $\argmin_{\vec{\theta}}L(\vec{\theta})$, and the ensuing parameters, along with the loss, are used to make predictions.

While VQAs and QML share some similarities, they also share some differences. Let us first discuss their similarities. First, in both frameworks, one trains a parametrized quantum circuit. This requires choosing an \textit{ansatz} for $U(\vec{\theta})$ and using a classical optimizer to train its parameters. As for their differences, in a VQA task as described in Definition~\ref{def:VQA} and Eq.~\eqref{eq:cost}, the input state to the parametrized quantum circuit $U(\vec{\theta})$ is usually an easy-to-prepare state $\ket{\psi_0}$ such as  the all-zero state, or some physically motivated product state (e.g.,  the Hartree-Fock state in quantum chemistry~\cite{cerezo2020variationalreview,romero2018strategies}). On the other hand, in a QML task as in Definition~\ref{def:QML} and Eq.~\eqref{eq:loss}, the input states to $U(\vec{\theta})$ are taken from the dataset $\mathcal{S}$, and thus can be extremely complex quantum states (see Fig.~\ref{fig:HEA}(a)).

\subsection{Hardware Efficient Ansatz}

As previously mentioned, one of the most important aspects of VQAs and QML models is the choice of ansatz for $U(\vec{\theta})$.  Without loss of generality, we assume that the parametrized quantum circuit is expressed as 
\begin{equation}\label{HEAequation}
    U(\vec{\theta}) = \prod_l e^{-i\theta_lH_l} V_l  \,,
\end{equation}
where the $\{V_l\}$ are some unparametrized unitaries, $\{H_l\}$ are traceless Pauli operators, and where $\vec{\theta}=(\theta_1,\theta_2,\ldots)$. While recently the field of ansatz design has seen a tremendous amount of interest, here we will focus on the HEA, one of the most widely used ansatzes in the literature. Originally introduced in Ref.~\cite{kandala2017hardware}, the term HEA is a generic name commonly reserved for ansatzes that are aimed at reducing the circuit depth by choosing gates $\{V_l\}$ and generators $\{H_l\}$ from a native gate alphabet determined from the connectivity and interactions to the specific quantum computer being used.

As shown in Fig.~\ref{fig:HEA}(b), throughout this work we will consider the most depth-frugal instantiation of the HEA: the one-dimensional alternating layered HEA. Here, one assumes that the physical qubits in the hardware are organized in a chain, where the $i$-th qubit can be coupled with the $(i-1)$-th and $(i+1)$-th. Then, at each \textit{layer} of the circuit one connects each qubit with its nearest neighbors in an alternating, brick-like, fashion. We will denote as $D$ the \textit{depth}, or the number of layers, of the ansatz. This type of alternating-layered HEA exploits the native connectivity of the device to maximize the number of operations at each layer while preventing qubits to idle. For instance, alternating-layered HEAs are extremely well suited for the IBM quantum hardware topology where only nearest neighbor qubits are directly connected (see e.g. Ref.~\cite{IBMQ16}).  We note that henceforth when we use the term HEA, we will refer to the alternating-layered ansatz of Fig.~\ref{fig:HEA}.

\section{Trainability of the HEA}\label{Sec: trainabilityHEA}

\subsection{Review of the literature}

In recent years, several results about the non-trainability of VQAs/QML have been pointed out~\cite{mcclean2018barren,cerezo2020cost,sharma2020trainability,thanasilp2021subtleties,holmes2021connecting,arrasmith2021equivalence,pesah2020absence,uvarov2020barren,marrero2020entanglement,patti2020entanglement}. In particular, it has been shown that quantum landscapes can exhibit the barren plateau phenomenon, which is nowadays considered to be  one of the most challenging bottlenecks for trainability of these hybrid models. We say that the cost function exhibits a barren plateau if, for the cost, or loss function, the optimization landscape becomes exponentially flat with the number of qubits. When this occurs, an exponential number of measurement shots are required to resolve and determine a cost-minimizing direction. In practice, the exponential scaling in the precision due to the barren plateaus erases the potential quantum advantage, as the VQA or QML scheme will have complexity comparable to the exponential scaling of classical algorithms. 

Being more concrete, let $f(\vec{\theta})= C(\vec{\theta}), L_{s}(\vec{\theta})$, i.e., either the cost function $C(\vec{\theta})$ of a VQAs, or the $s$-th term $L_{s}(\vec{\theta})$ in the loss function  of a QML settings. For simplicity of notation, we will omit the ``$s$'' sub-index of $O_s$ and $\psi_s$ when $f(\vec{\theta})=  L_{s}(\vec{\theta})$. In a barren plateau, there are two types of concentration (or flatness) notions that have been explored: deterministic concentration (all landscape is flat) and probabilistic (most of the landscape is flat). Let us first define the deterministic notion of concentration:
\begin{definition}(Deterministic concentration)\label{def: deterministic}
Let the trivial value of the cost function be $f_{trv}\coloneqq \frac{1}{2^n}\tr[O]$. Then the function $f(\vec{\theta})$ is $\epsilon-$concentrated iff $\exists \, \epsilon>0$ s.t. for any $\vec{\theta}$
\be
|f(\vec{\theta})-f_{trv}|\le \epsilon\,.
\ee
\end{definition}

The above definition puts forward a necessary condition for trainability. It is clear that if $f(\vec{\theta})$ is $\epsilon$-concentrated, then $f(\vec{\theta})$ must be resolved within an error that scales as $\sim \epsilon$, i.e., one must use  $\sim \epsilon^{-2}$ measurement shots to estimate $f(\vec{\theta})$. Thus, we define a VQA/QML model to be trainable if $\epsilon$ vanishes no faster than polynomially with $n$ ($\epsilon\in\Omega(1/\poly(n))$). Conversely, if $\epsilon\in\mathcal{O}(2^{-n})$, one requires an exponential number of measurement shots to resolve the quantum landscape, making the model non-scalable to a higher number of qubits. Deterministic concentration was shown in Refs.~\cite{wang2020noise,franca2020limitations}, which study the performance of VQA and QML models in the presence of quantum noise and  prove that $|f(\vec{\theta})-f_{trv}|\in\mathcal{O}(q^D)$, where $0<q<1$ is a parameter that characterizes the noise. Using the results therein, it can be shown that if the depth $D$ of the HEA is $D\in\mathcal{O}(\poly(n))$, then the  noise acting through the circuit leads to an exponential concentration around the trivial value $f_{trv}$.

Let us now consider the following definition of probabilistic concentration:
\begin{definition}[Probabilistic concentration]\label{Def: probabilisticConc}
Let $\aver{\cdot}_{\vec{\theta}}$ be the average with respect to the parameters $\theta_{1}\in\Theta_1,\theta_{2}\in\Theta_2,\ldots$, for a set of parameter domains $\{\vec{\Theta}\}$. Then the function $f(\vec{\theta})$ is probabilistic concentrated if for any $\vec{\theta}'$
\be
\mathbb{E}_{\vec{\theta}}[f(\vec{\theta})-f(\vec{\theta}')] \leq \epsilon\,,
\ee
where the average is taken over the domains $\{\vec{\Theta}\}$.
\end{definition}
Here we make an important remark on the connection between probabilistic concentration and barren plateaus. The  barren plateau phenomenon, as initially formulated in Ref.~\cite{mcclean2018barren} indicates that the cost function gradients are concentrated, i.e., that $\Var[\partial_{\nu}f(\vec{\theta})]\equiv\aver{[\partial_{\nu}f(\vec{\theta})]^{2}}_{\vec{\theta}}-\aver{\partial_{\nu}f(\vec{\theta})}^{2}_{\vec{\theta}}\le \epsilon$, where $\partial_{\nu}f(\vec{\theta})\coloneqq \partial f(\vec{\theta})/\partial\theta_\nu $. However one can prove that probabilistic cost concentration implies probabilistic gradient concentration, and vice-versa~\cite{arrasmith2021equivalence}. 
According to Definition~\ref{Def: probabilisticConc}, we can again see that  if $\epsilon\in\mathcal{O}(2^{-n})$, one requires an exponential number of measurement shots to navigate through the optimization landscape. As shown in Ref.~\cite{mcclean2018barren}, such probabilistic concentration can occur if  the  depth is $D\in \mathcal{O}(\poly(n))$, as at the ansatz becomes a $2$-design~\cite{brandao2016local,harrow2018approximate,mcclean2018barren}.

From the previous, we know that deep HEAs with $D\in \mathcal{O}(\poly(n))$ can exhibit both deterministic cost concentration (due to noise), but also probabilistic cost concentration (due to high expressibility~\cite{holmes2021connecting}). However, the question still remains open of whether HEA can avoid barren plateaus and cost concentration with sub-polynomial depths. This question was answered in Ref.~\cite{cerezo2020cost}. Where it was shown that HEAs can avoid barren plateaus and have trainability guarantees if two necessary conditions are met: \textit{locality and shallowness}. In particular, one can prove that if $D\in\mathcal{O}(\log(n))$, then measuring  \textit{global} operators -- i.e., $O$ is a sum of operators acting non-identically on every qubit -- leads to barren plateaus, whereas measuring \textit{local} operators -- i.e. $O$ is a sum of operators acting (at most) on $k$ qubits, for $k \in\mathcal{O}(1)$ -- leads to gradients that vanish only polynomially in $n$. 

\subsection{A new source for untrainability}

The discussions in the previous section provide a sort of recipe for avoiding expressibility-induced probabilistic concentration (see Definition~\ref{Def: probabilisticConc}), and noise-induced deterministic concentration: \textit{Use local cost measurement operators and keep the depth of the quantum circuit shallow enough}.

Unfortunately, the previous is still not enough to guarantee trainability. As there are other sources of untrainability which are usually less explored. To understand what those are, we will recall a simplified version of the main result in Theorem 2 of Ref.~\cite{cerezo2020cost}. First, let $O$ act non-trivially only on two adjacent qubits, one of them being the $\lfloor \frac{n}{2}\rfloor$-th qubit, and let us study the partial derivative $\partial_{\nu}f(\vec{\theta})$ with respect to a parameter in the last gate acting before $O$ (see Fig.~\ref{fig:lightcone}). The variance  of $\partial_{\nu}f(\vec{\theta})$   is lower bounded as~\cite{cerezo2020cost} 
\be
G_{n}(D,\ket{\psi})\le \Var[\partial_{\nu}f(\vec{\theta})]\,,
\label{lowerboundCerezo0}
\ee
where we recall that $D$ is the depth of the HEA, $\ket{\psi}$ is the input state, and with 
\be
G_{n}(D,\ket{\psi}) \!=\!\left(\frac{2}{5}\right)^{\scriptscriptstyle 2D}\!\!\!\frac{2 }{225} \,\eta(O)\!\!\sum_{\substack{\scriptscriptstyle k,k^{\prime}= n/2-D\\ \scriptscriptstyle k^{\prime}>k}}^{\scriptscriptstyle n/2+D}\!\!\!\eta(\psi_{k,k^{\prime}})\,.
\label{lowerboundCerezo}
\ee

Here $\eta(M)=\norm{M-\Tr[M]\frac{\mathds{1}}{d_M}}_2$, is the Hilbert-Schmidt distance between $M$ and $\Tr[M]\frac{\mathds{1}}{d_M}$,  where $d_M$ is the dimension of the matrix $M$. Moreover, here we defined $\psi_{k,k^\prime}=\Tr_{1,\ldots,k-1,k^{\prime}+1,\ldots,n}[\st{\psi}]$ as the reduced density matrix on the qubits with index $i\in[k,k^{\prime}]$. As such, $\psi_{k,k^\prime}$  correspond to the reduced states of all possible combinations of adjacent qubits in the \textit{light-cone} generated by $O$ (see Fig.~\ref{fig:lightcone}). Here, by light-cone we refer to the set of qubit indexes that are causally related to $O$ via $U(\vec{\theta})$, i.e., the set of indexes over which $U^\dagger(\vec{\theta})O U(\vec{\theta})$ acts non-trivially.

\begin{figure}[t]
    \centering
    \includegraphics[width=.9\linewidth]{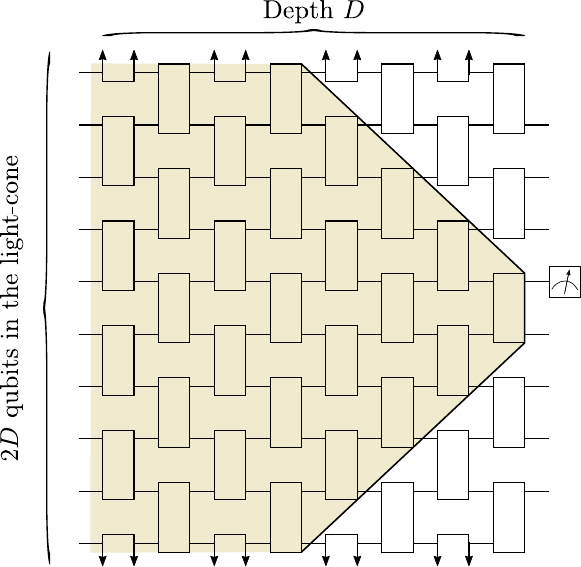}
    \caption{The sketch shows the light-cone of a local measurement operator at the end of a shallow HEA. Here we can see how the support of the local operator grows with the depth $D$ of the HEA. Since the HEA gates act on neighboring qubits, the support increases no more than $2D$.  }
    \label{fig:lightcone}
\end{figure}

Equation~\eqref{lowerboundCerezo} provides the necessary condition to guarantee trainability, i.e., to ensure that the gradients do not vanish exponentially. First, one recovers the condition on the HEA that $D\in\mathcal{O}(\log (n))$. However, a closer inspection of the above formula reveals that both the initial state $\ket{\psi}$ and the measurement operator $O$ also play a key role. Namely one needs that $O$, as well as \textit{any} of the reduced density matrices of $\ket{\psi}$ an \textit{any} set adjacent qubits in the light-cone, to not be close (in Hilbert-Schmidt distance) to the (normalized) identity matrix. This is due to the fact that if $\eta(O)\in\mathcal{O}(2^{-n})$, or if $\eta(\psi_{k,k^{\prime}})\in\mathcal{O}(2^{-n})$ $\forall k,k^\prime$, then $G_{n}(D,\ket{\psi})\in\mathcal{O}(2^{-n})$ and the trainability guarantees are lost (the lower bound in Eq.~\eqref{lowerboundCerezo0} becomes trivial).

The previous results highlight that one should pay close attention to the measurement operator and the input states. Moreover, these results make intuitive sense as they say that extracting information by measuring an operator $O$ that is exponentially close to the identity will be exponentially hard. Similarly, training an ansatz with local gates on a state whose marginals are exponentially close to being maximally mixed will be exponentially hard. 

Here we remark that in a practical scenario of interest, one does not expect $O$ to be exponentially close to the identity. For a VQA, one is interested in finding the ground state of $O$ (see Eq.~\eqref{eq:cost}), and as such, it is reasonable to expect that $O$ is non-trivially close to the identity~\cite{cerezo2020variationalreview,bharti2021noisy}. Then, for QML there is additional freedom in choosing the measurement operators $O_{s}$ in~\eqref{eq:loss}, meaning that one simply needs to choose an operator with non-exponentially vanishing support in non-identity Pauli operators. 

In the following sections, we will take a closer look at the role that the input state can have in the trainability of shallow-depth HEA.

\section{Entanglement and information scrambling}

Here we will briefly recall two fundamental concepts: that of states satisfying an area law of entanglement, and that of states satisfying a volume law of entanglement. Then, we will relate the concept of area law of entanglement with that of scrambling. 

First, let us rigorously define what we mean by area and volume laws of entanglement.
\begin{definition}[Volume law vs. area law]\label{def: volumelaw}
Let $|\psi\rangle$ be a state in a bipartite Hilbert space $\mathcal{H}_{\Lambda}\otimes \mathcal{H}_{\bar{\Lambda}}$. Let $\Lambda$ be a subsystem composed of $|\Lambda|$  qubits, and let $\bar{\Lambda}$ be its complement set. Let $\psi_{\Lambda}=\Tr_{\bar{\Lambda}}[\st{\psi}]$ be the reduced density matrix on $\Lambda$. Then, the state $\ket{\psi}$ possesses volume law for the entanglement within $\Lambda$ and $\bar{\Lambda}$ if 
\be
|\Lambda|-S(\psi_{\Lambda})\in\mathcal{O}\left(\frac{1}{2^{|\bar{\Lambda}|-|\Lambda|}}\right)
\label{volumelawdef}\,,
\ee
where $S(\rho)=-\Tr[\rho\log(\rho)]$ is the entropy of entanglement. Conversely, the state possesses area law for the entanglement within $\Lambda$ and $\bar{\Lambda}$ if
\be
|\Lambda|-S(\psi_{\Lambda})\in\Omega\left(\frac{1}{\poly(|\Lambda|)}\right)\,.
\ee
\end{definition}
Note that the above definition of area vs. volume law of entanglement is nonstandard. In particular, it is completely agnostic to the geometry on which the given state resides. However, as we will show, it is the relevant one to look at for the purpose of this work.

From the definition of volume law of entanglement, the concept of scrambling of quantum information can be easily defined; the information contained in $\ket{\psi}$ is said to be scrambled throughout the system if the state $\ket{\psi}$ follows a volume law for the entropy of entanglement according to Definition~\ref{def: volumelaw} across any bipartition  such that $|\Lambda|\in\mathcal{O}(\log (n))$. 

Here we further recall that an information-theoretic measure of the quantum information that can be extracted by a subsystem $\Lambda$ is
\be
\mathcal{I}_{\Lambda}(\psi)\coloneqq \norm{\psi_{\Lambda}-\frac{\mathds{1}_{\Lambda}}{2^{|\Lambda|}}}_1\,,
\ee
which quantifies the maximum distinguishability between the reduced density matrix $\psi_{\Lambda}$ and the maximally mixed state $\mathds{1}_{\Lambda}/2^{|\Lambda|}$. To see this, let $O_{\Lambda}$ be a local operator acting non-trivially on $\Lambda$, whose maximum eigenvalue is one, then $|\braket{\psi|O_{\Lambda}|\psi}-2^{-n}\Tr[O_{\Lambda}]|\le \mathcal{I}_{\Lambda}(\psi)$. Thus, if $\mathcal{I}_{\Lambda}$ is exponentially small, then the measurement of the local operator $O_{\Lambda}$ is not able to efficiently reveal any information contained in $\ket{\psi}$. We thus lay down the following formal definition of information scrambling:
\begin{definition}[Information scrambling]\label{def-scrambling}
The quantum information in $\ket{\psi}$ is scrambled iff for any subset of qubits $\Lambda$ such that $|\Lambda|\in\mathcal{O}(\log (n))$:
\be
\mathcal{I}_{\Lambda}(\psi)\in \mathcal{O}(2^{-cn})\,.
\ee
for some constant $c>0$.
\end{definition}
The definition of scrambling of quantum information easily follows from the definition of volume law for entanglement in Definition~\ref{def: volumelaw}. Indeed, given a subsystem $\Lambda$, one has the following bound
\be
\mathcal{I}_{\Lambda}(\psi)\le \sqrt{2(|\Lambda|-S(\psi_{\Lambda}))}\,,
\ee
and thus, if $\ket{\psi}$ follows a volume law of entanglement according to Definition~\ref{def: volumelaw} for any bipartition $\Lambda\cup\bar{\Lambda}$ with $|\Lambda|\in\mathcal{O}(\log (n))$, one indeed has the exponential suppression of the information contained in $\Lambda$, i.e., $\mathcal{I}_{\Lambda}(\psi)\in \mathcal{O}(2^{-n})$. This motivates us to propose the following alternative definition for states following volume and area law of entanglement 
\begin{definition}[Volume law vs. area law]\label{def: volumelaw2}
Let $|\psi\rangle$ be a state in a bipartite Hilbert space $\mathcal{H}_{\Lambda}\otimes \mathcal{H}_{\bar{\Lambda}}$. Let $\Lambda$ be a subsystem composed of $|\Lambda|$  qubits, and let $\bar{\Lambda}$ be its complement set. Let $\psi_{\Lambda}=\Tr_{\bar{\Lambda}}[\st{\psi}]$ be the reduced density matrix on $\Lambda$. Then the state $\ket{\psi}$ possesses volume law for the entanglement within $\Lambda$ and $\bar{\Lambda}$ if 
\be
\mathcal{I}_{\Lambda}(\psi)\in \mathcal{O}\left(\frac{1}{2^{cn}}\right)\,,
\ee
for some $c>0$. Conversely, the state possesses area law for the entanglement within $\Lambda$ and $\bar{\Lambda}$ if
\be
\mathcal{I}_{\Lambda}(\psi)\in \Omega\left(\frac{1}{\poly(n)}\right)\,.
\ee
\end{definition}

\section{HEA and the role of entanglement}
Let us consider a VQA or QML task from Definitions~\ref{def:VQA}--\ref{def:QML}, where the ansatz for the parametrized quantum circuit $U(\vec{\theta})$ is a shallow  HEA with depth $D\in\mathcal{O}(\log(n))$ (see Fig.~\ref{fig:HEA}(b)). Moreover,  let the function $f(\vec{\theta})=C(\vec{\theta}), L_{s}(\vec{\theta})$, be either the cost or loss function in Eqs.~\eqref{eq:cost} and~\eqref{eq:loss}.

\subsection{HEA and volume law of entanglement}
As we show in this section, a shallow HEA will be untrainable if the input state satisfies a volume law of entanglement according to Definition~\ref{def: volumelaw}. Before going deeper into the technical details, let us sketch the idea behind our statement, with the following warm-up example.

\subsubsection{A toy model}
Let us consider for simplicity the case when $O$ is a local operator acting non-trivially  on a single qubit, and let us recall that  $f(\vec{\theta})=\Tr[U(\vec{\theta})\st{\psi}U^{\dag}(\vec{\theta})O]$. In the Heisenberg picture we can interpret $f(\vec{\theta})$ as the expectation value of the backwards-in-time evolved operator $O(\vec{\theta})=U^{\dag}(\vec{\theta})O U(\vec{\theta})$ over the initial state $\ket{\psi}$. Thanks to the brick-like structure of the HEA, one can see from a simple geometrical argument that the operator $O(\vec{\theta})$ will act non-trivially only on a set $\Lambda$ containing (at most) $2D$ qubits (see Fig.~\ref{fig:lightcone}, and also see below for the rigorous proof). Thus, we can compute $f(\vec{\theta})$ as 
\be
f(\vec{\theta})=\Tr_{\Lambda}[\Tr_{\bar{\Lambda}}[\st{\psi}] O(\vec{\theta})] \; ,
\ee
where $\bar{\Lambda}$ is the complement set of $\Lambda$, and where we assume $|\Lambda|\ll |\bar{\Lambda}|$. Since the HEA is shallow, the cost function is evaluated by tracing out the majority of qubits, i.e. $|\bar{\Lambda}|\sim n$. If the input state $\ket{\psi}$ is highly entangled, since $|\Lambda|\ll |\bar{\Lambda}|$, and thanks to the monogamy of entanglement~\cite{coffman2000distributed}, we can assume that there is a subset of $\bar{\Lambda}$, say $\Lambda^{\prime}$, maximally entangled with $\Lambda$, i.e., $\ket{\psi}= \sqrt{\frac{1-\epsilon^2}{2^{|\Lambda|}}} \sum_{i=1}\ket{i_\Lambda}\otimes\ket{i_{\Lambda^\prime}}\otimes \ket{rest}+\epsilon\ket{\phi}$, where $\epsilon\ll1$ and $\ket{\phi}$ being orthogonal to the rest. Neglecting terms in $\epsilon^2$, and choosing $\norm{O}_{\infty}=1$, this results in a function $2\epsilon$-concentrated around its trivial value $f_{trv}$
\be
|f(\vec{\theta})-f_{trv}|\le 2\epsilon\,.
\ee
The previous shows that a highly entangled state, such as the one presented above which satisfies a volume law-of entanglement, will lead to a landscape that exhibits a deterministic exponential concentration according to Definition~\ref{def: deterministic}.

\subsubsection{Formal statement}\label{Sec: formal}
In this section, we will present a deterministic concentration result. To begin, let us introduce the following definition:
\begin{definition}[Support of a Pauli operator]
Let $P$ be a Pauli operator, we define the support $\supp(P)$ as the ordered set of natural numbers $q_i$ labeling the qubits on which $P$ acts non-trivially.
\be
\supp(P)=\{q_{1},\ldots, q_{S_i}\,|\, q_{1}< \ldots<q_{S}, q_{k}\in [1,n]\}
\ee
with $S$ the number of qubits on which $P$ acts non-trivially.
\end{definition}
We are finally ready to state a deterministic concentration result based on the information-theoretic measure $\mathcal{I}_{\Lambda}(\psi)$ for HEA circuits and for    $f(\vec{\theta})= C(\vec{\theta}), L_{s}(\vec{\theta})$ being the VQA cost function or the QML loss function. 

\begin{theorem}[Concentration and measurement operator support]\label{th1}
Let $U(\vec{\theta})$ be HEA with depth $D$, and $O=\sum_{i}c_iP_i$, with $P_i$ Pauli operators. Now, let $\Lambda_i\coloneqq \supp(U^{\dag}(\vec\theta)P_{i}U(\vec\theta))$:
\be\label{eq:bound}
|f(\vec{\theta})-f_{trv}|\le \sum_{i}|c_i| \mathcal{I}_{\Lambda_i}(\psi)
\ee
moreover, $\sum_{i}|c_i|\le (3n)^{\max_i|\supp(P_i)|}\norm{O}_{\infty}$. Let $\Lambda\coloneqq \max_{i}\Lambda_i$. Then, the following bound on the size of $\Lambda$ holds:
\be
|\Lambda|\le 2D\times\max_{i}|\supp(P_i)|\label{sizelambda}\,.
\ee
\end{theorem}
See App.~\ref{App: proofth1} for the proof.

Let us discuss the implications of Theorem~\ref{th1}. First, we find that the difference between the training function $f(\vec{\theta})$, and its trivial value $f_{trv}$ depends on the information-theoretic measure of information scrambling $\mathcal{I}_{\Lambda}(\psi)$ for $|\Lambda|=\max_{i}|\Lambda_i|$ (see Definition~\ref{def-scrambling}). From Eq.~\eqref{sizelambda} it is clear that the size of $\Lambda$ is determined by two factors: $(i)$ the depth of the circuit $D$, and $(ii)$ the locality of the operator $O$. As soon as either the depth $D$ or $\max_i\supp(P_i)$ starts scaling with the number of qubits $n$, the bound in Eq.~\eqref{eq:bound} becomes trivial as one can obtain information by measuring a large enough subsystem with $|\Lambda|\in\Theta(n)$. However, as explained above in Sec.~\ref{Sec: trainabilityHEA}, we already know that this regime is precluded as the necessary requirements to ensure the trainability of the HEA (and thus to ensure trainability of the VQA/QML model) are: $(i)$ the depth of the HEA circuit must not exceed $\mathcal{O}(\log (n))$, and $(ii)$ the operator $O$ must have local support on at most $\mathcal{O}(\log (n))$ qubits. Hence, the trainability of the model is solely determined by the scaling of $\mathcal{I}_{\Lambda}(\psi)$.

From Theorem~\ref{th1}, we can derive the following corollaries
\begin{corollary}\label{cor2}
Let the depth $D$ of the HEA be $D\in\mathcal{O}(\log (n))$, and let $\max_{i}|\supp(P_i)|\in\mathcal{O}(\log (n))$. Then, if $\ket{\psi}$ satisfies a volume law, or alternatively if  the information contained in $\ket{\psi}$ is scrambled, i.e., if $\mathcal{I}_{\Lambda}(\psi)\in\mathcal{O}(2^{-cn})$ for some $c>0$, and if $\norm{O}_{\infty}\in\Omega(1)$, then:
\be
|f(\vec{\theta})-f_{trv}|\in \mathcal{O}(2^{-nc/2})\,.
\ee
\end{corollary}
Here we can see that if the information contained in $\ket{\psi}$ is too scrambled throughout the system, one has deterministic exponential concentration of cost values according to  Definition~\ref{def: deterministic}.

Theorem~\ref{th1} puts forward another important necessary condition for trainability and to avoid deterministic concentration: \textit{the information in the input state must not be too scrambled throughout the system}. When this occurs, the information in $\ket{\psi}$ cannot be accessed by local measurements, and hence one cannot train the shallow depth HEA. 

At this point, we ask the question of how typical is for a state to contain information scrambled throughout the system, hidden in non-local degrees of freedom, and resulting in $\mathcal{I}_{\Lambda}(\psi)\in\mathcal{O}(2^{-n})$. To answer this question, we use tools of the Haar measure and show that for the overwhelming majority of states, the information cannot be accessed by local measurements as their information is too scrambled.
\begin{corollary}\label{application2}
Let $\ket{\psi}\sim \mu_{Haar}$ be a Haar random state, $D\in\mathcal{O}(\log (n))$, and $\max_i\supp(P_i)\in\mathcal{O}(\log (n))$. Here $\mu_{Haar}$ is the uniform Haar measure over the states in the Hilbert space.  Then $\mathcal{I}_{\Lambda}(\ket{\psi})\le 2^{|\Lambda|-n/3}$, with overwhelming probability $\ge 1-e^{-c2^{n/3-1}+(2n+1)\ln 2}$ with $c=(18\pi^3)^{-1}$, and thus
\be
|f(\vec{\theta})-f_{trv}|\in \mathcal{O}(2^{-n/6})\,.
\ee

\end{corollary}
See App.~\ref{App: proofth1} for a proof.
Note that henceforth, we will refer to \emph{overwhelming probability} as a probability $1$ up to a exponentially (in $n$, the size of the system) decaying correction.
In many tasks, multiple copies of a quantum state $\ket{\psi}$ are used to predict important properties, such as entanglement entropy~\cite{Abanin2012meas,foulds2020controlled,beckey2021computable}, quantum magic~\cite{leone2021renyi,oliviero2022measuring,haug2022magic}, or state discrimination~\cite{kang2019implementation}. Thus, it is worth asking whether a function of the form $f(\vec{\theta})=\Tr[U(\vec{\theta})\st{\psi}^{\otimes 2}U^{\dag}(\vec{\theta})O]$ can be trained when $U(\vec{\theta})$ is a shallow HEA acting on $2n$ qubits. In the following corollary, we prove that for the overwhelming majority of states, $\mathcal{I}_{\Lambda}(\psi^{\otimes 2})\in\mathcal{O}(2^{-n})$, and one has deterministic concentration according to Definition~\ref{def: deterministic} even if one has access to two copies of a quantum state. 
\begin{corollary}\label{app2copies}
Suppose one has access to $2$ copies of a Haar random state $\ket{\psi}$ and one computes the function $f(\vec{\theta})=\Tr[U(\vec{\theta})\st{\psi}^{\otimes 2}U^{\dag}(\vec{\theta})O]$. Let $D\in\mathcal{O}(\log (n))$ be the depth of a HEA $U(\vec{\theta})$, and $\max_i\supp(P_i)\in\mathcal{O}(\log (n))$. Then $\mathcal{I}_{\Lambda}(\ket{\psi}^{\otimes 2})\le 2^{|\Lambda|-n/3}$,  with overwhelming probability $\ge 1-e^{-c2^{n/3-4}+(2n+1)\ln 2}$, where $c=(18\pi^3)^{-1}$, and one has:
\be
|f(\vec{\theta})-f_{trv}|\in \mathcal{O}(2^{-n/6})\,.
\ee
\end{corollary}
See App.~\ref{App: proofth1} for a proof. Note that the generalization to more copies is straightforward.

The above results show us that there is indeed a no-free-lunch for the shallow HEA. The majority of states in the Hilbert space, follow a volume law for the entanglement entropy and thus have quantum information hidden in highly non-local degrees of freedom, which cannot be accessed through local measurement at the output of a shallow HEA.

\subsection{HEA and area law of entanglement}\label{Sec: arealaw}

The previous results indicate that shallow HEAs are untrainable for states with a volume law of entanglement, i.e., they are untrainable for the vast majority of states. The question still remains of whether shallow HEA can be used if the input states follow an area law of entanglement as in Definition~\ref{def: volumelaw}.   Surprisingly, we can show that in this case there is no concentration, as  the following result holds 
\begin{theorem}[Anti-concentration of expectation values]\label{theore:anticoncentration}
Let $U(\vec{\theta})$ be a shallow HEA with depth $D\in\mathcal{O}(\log(n))$ where each local two-qubit gate forms a $2$-design on two qubits. Then, let $O=\sum_{i}c_iP_i$ be the measurement composed of, at most, polynomially many traceless Pauli operators  $P_i$ having support on at most two neighboring qubits, and where $\sum_i c_i^2\in\mathcal{O}(\poly(n))$. If the input state follows an area law of entanglement,  for any set of parameters $\vec{\theta_B}$ and $\vec{\theta_A}=\vec{\theta_B}+\hat{\vec{e}}_{AB}l_{AB}$ with $l_{AB}\in\Omega(1/\poly(n))$, then
\be
 \Var_{\vec{\theta}_B}[f(\vec{\theta}_A)-f(\vec{\theta}_B)]\in\Omega\left(\frac{1}{\poly(n)}\right)\,.
\label{lowerboundCerezo2}
\ee
\end{theorem}
See App.~\ref{App: Theorem 2} for the proof.

Theorem~\ref{theore:anticoncentration} shows that if the input states to the shallow HEA follow an area law of entanglement, then the function $f(\vec{\theta})$ anti-concentrates. That is, one can expect that the loss function values will differ (at least polynomially) at sufficiently different points of the landscape. This naturally should imply that the cost function does not have barren plateaus or exponentially vanishing gradients. In fact, we can prove this intuition to be true as it can be formalized in the following result.
\begin{prop}\label{eq:prop-anticoncentration}
Let $f(\vec{\theta})$ be a VQA cost function or a QML loss function where $U(\vec{\theta})$ is  a shallow HEA with depth $D\in\mathcal{O}(\log(n))$. If the values of $f(\vec{\theta})$ anti-concentrate according to Theorem~\ref{theore:anticoncentration} and Eq.~\eqref{lowerboundCerezo2}, then $\Var_{\vec{\theta}}[\partial_\nu f(\vec{\theta})]\in\Omega(1/\poly(n))$ for any $\theta_\nu\in\vec{\theta}$ and the loss function does not exhibit a barren plateau. Conversely, if $f(\vec{\theta})$ has no barren plateaus, then the cost function values anti-concentrate as in Theorem~\ref{theore:anticoncentration} and Eq.~\eqref{lowerboundCerezo2}.
\end{prop}
See App.~\ref{App: Theorem 2} for the proof.

Taken together, Theorem~\ref{theore:anticoncentration} and Proposition~\ref{eq:prop-anticoncentration} suggest that shallow HEAs are ideal for processing states with area law of entanglement, as the loss landscape is immune to barren plateaus. Evidently, the previous shallow HEAs are capable of achieving a quantum advantage. However, determining whether a quantum advantage is feasible or not, for such ansatzes is beyond the scope of this work (as it requires a detailed analysis of properties beyond the absence of barren plateaus such as quantifying the presence of local minima) we can still further identify scenarios where a quantum advantage could potentially exist. 

First, let us rule out certain scenarios where a provable quantum advantage will be unlikely. These correspond to cases where the input state $\ket{\psi}$ satisfies an area law of entanglement but also admits an efficient classical representation~\cite{orus2014practical,schollwock2011density,verstraete2008matrix,schuch2008entropy,verstraete2004renormalization}.  The key issue here is that if the input state admits a classical decomposition, then the expectation value  $f(\vec{\theta})$ for $U(\vec{\theta})$ being a shallow HEA can be efficiently classically simulated~\cite{kottmann2022optimized}. For instance, one can readily show that the following result holds.
\begin{prop}[Cost of classically computing $f(\vec{\theta})$]\label{theo:classical}
Let $U(\vec{\theta})$ be an alternating layered HEA of depth $D$, and $O=\sum_{i}c_iP_i$. Let $\ket{\psi}$ be an input stat that admits a Matrix Product State~\cite{orus2014tensornetwork} (MPS) description with bond-dimension $\chi$. Then, there exists a classical algorithm that can estimate $f(\vec{\theta})$ with a complexity which scales as $\mathcal{O}((\chi \cdot4^D)^3)$.
\end{prop}

The proof of the above proposition can be found in~\cite{orus2014tensornetwork}. From the previous theorem, we can readily derive the following corollary. 
\begin{corollary}\label{coro:classical}
Shallow depth HEAs with depth $D\in \mathcal{O}(\log(n))$, and with an input state with a bond-dimension $\chi\in \mathcal{O}(\poly(n))$ can be efficiently classically simulated with a complexity that scales as $\mathcal{O}(\poly(n))$.
\end{corollary}

Note that Proposition~\ref{theo:classical} and its concomitant Corollary~\ref{coro:classical} do not preclude the possibility that shallow HEA can be useful even if the input state admits an efficient classical description. This is due to the fact that, while requiring computational resources that scale polynomially with $n$ (if $\chi$ is at most polynomially large with $n$), the order of the polynomial can still lead to prohibitively large (albeit polynomially growing) computational resources. Still, we will not focus on discussing this fine line, instead, we will attempt to find scenarios where a quantum advantage can be achieved.

In particular, we highlight the seminal work of Ref.~\cite{ge2016area}, which indicates that  while states satisfying an area law of entanglement constitute just a very small fraction of all the states (which is expected from the fact that Haar random states --the vast majority of states-- satisfy a volume law),  the subset of such area law of entanglement states that admit an efficient classical representation is exponentially small. This result can be better visualized in Fig.~\ref{fig:advantage}. The previous gives hope that one can achieve a quantum advantage with  \textit{area law classically-unsimulable states}.

\begin{figure}
    \centering
    \includegraphics[width=.9\linewidth]{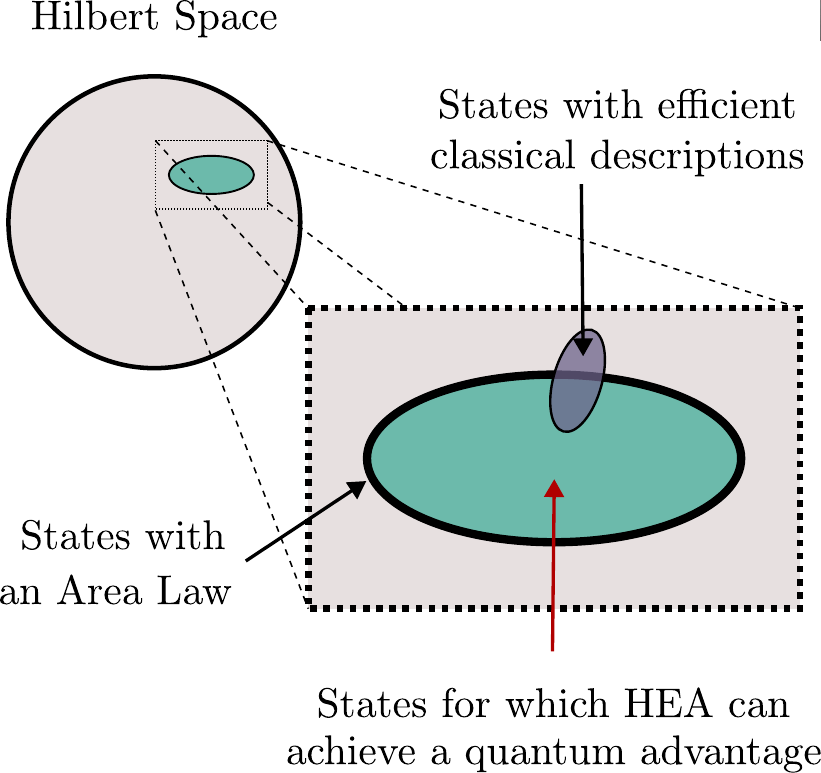}
    \caption{Schematic representation of the Hilbert space. The vast majority of states satisfy a volume law, and hence a shallow HEA cannot be used to extract information from them. From the set of states satisfying an area law, only a very small subset admits an efficient classical representation. For these states, the effect of a shallow HEA can be efficiently simulated. As such, there exists a Goldilocks regime where HEA can potentially be used to achieve a quantum advantage: non-classically-simulable area law states. }
    \label{fig:advantage}
\end{figure}

\begin{figure*}[t]
\centering
\includegraphics[width=\textwidth]{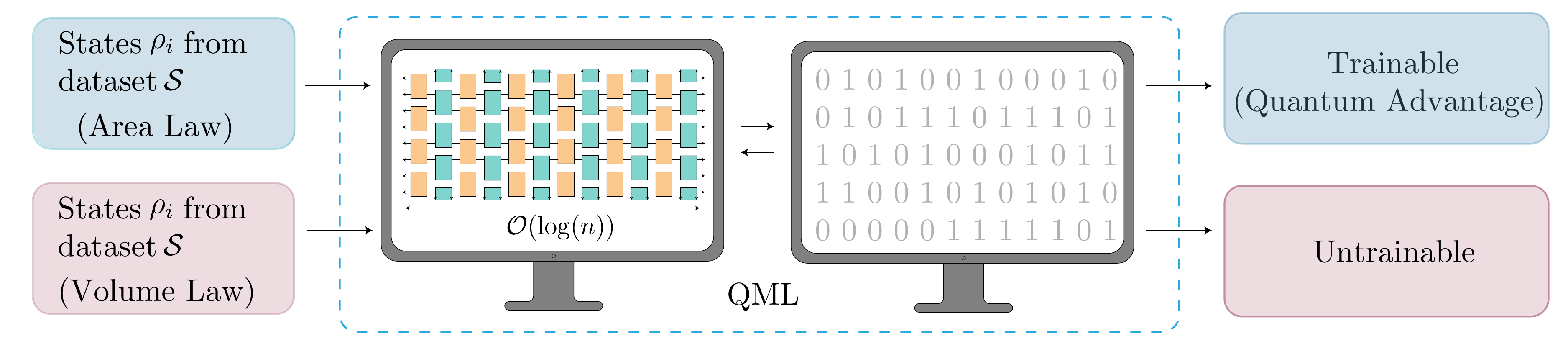}
\caption{Schematic representation of the trainability of two QML tasks with two different datasets $\mathcal{S}$. 
The first task is trainable since the dataset $\mathcal{S}$ is composed of states $\rho_i$ possessing an area law of entanglement. Conversely, the second task is untrainable being the dataset $\mathcal{S}$ composed of states $\rho_i$ possessing volume law of entanglement. We remark that the first task can enjoy a quantum advantage since not all the area-law states are classically simulable, see Ref.~\cite{ge2016area} and Fig.~\ref{fig:advantage}.}
\label{fig:fig4}
\end{figure*}

\section{Implications of our results}

Let us here discuss how our results can help identify scenarios where shallow HEA can be useful, and scenarios where they should be avoided.

\subsection{Implications to VQAs}

As indicated in Definition~\ref{def:VQA}, in a VQA one initializes the circuit to some easy-to-prepare fiduciary quantum state $\ket{\psi_0}$. For instance, in a variational quantum eigensolver~\cite{peruzzo2014variational} quantum chemistry application such an initial state is usually the un-entangled mean-field  Hartree-Fock state~\cite{romero2018strategies}. Similarly, when solving a combinatorial optimization task with the quantum optimization approximation algorithm~\cite{farhi2014quantum} the initial state is an equal superposition of all elements in the computational basis $\ket{+}^{\otimes n}$. In both of these cases, the initial states are separable, satisfy an area law, and admit an efficient classical decomposition. This means that while the shallow HEA will be trainable, it will also be classically simulable. This situation will arise for most VQA implementations as it is highly uncommon to prepare non-classically simulable initial states. From the previous, we can see that shallow HEA should likely be avoided in VQA implementations if one seeks to find a quantum advantage.  

\subsection{Implications to QML}

In a QML task according to Definition~\ref{def:QML}, one sends input states from some datasets into the shallow HEA. As shown in Fig.~\ref{fig:fig4}, the input states are problem-dependent, implying that the usability of the HEA  depends on the task at hand. Our results indicate that HEAs should be avoided when the input states satisfy a volume law of entanglement or when they follow an area law but also admit an efficient classical description. In fact, it is clear that while the HEA is widely used in the literature, most cases where it is employed fall within the cases where the HEA should be avoided~\cite{thanasilp2021subtleties}. As such, we expect that many proposals in the literature should be revised. However, the trainability guarantees pointed out in this work, narrow down the scenarios where the HEA should be used, and leave the door open for using shallow HEAs in QML tasks to analyze non-classically simulable area-law states. In the following section, we give an explicit example, based on state discrimination between area law states having no MPS decomposition, with a possible achievable quantum advantage.

\section{Random Hamiltonian discrimination}\label{sec: HamiltonianDiscrimination}

\subsection{General framework}

In this section, we present an application of our results in a QML setting based on Hamiltonian Discrimination. The QML problem is summarized as follows: the data contains states that are obtained by evolving an initial state either by a general Hamiltonian or by a Hamiltonian possessing a given symmetry. The goal is to train a QML model to distinguish between states arising from these two evolutions. In the example below, we show how the role of entanglement governs the success of the QML algorithm.

Let us begin by formally stating the problem. Consider two Hamiltonians $H_{G}, H_{S}$, and a local operator $S$ being a symmetry for $H_{S}$, i.e. $[H_{S},S]=0$. Let $\mathcal{V}_S\coloneqq \operatorname{span}\{\ket{z}\in\mathbb{C}^{2\otimes n}\,|\, S\ket{z}=\ket{z}\}$, i.e. the linear subspace filled by eigenvectors of $S$. Let us assume that $\operatorname{dim}(\mathcal{V}_S)\ge 2$. Then, we build the dataset following Algorithm~\ref{alg1}.
\begin{algorithm}[H]
\caption{Build the dataset $\mathcal{S}$}\label{alg1}
\algsetup{indent=2em}
\begin{algorithmic}[1]
\STATE $\mathcal{S}=\emptyset$
\STATE s=0
\WHILE{$s\leq N$}
\STATE take $\ket{z_s}\in\mathcal{V}_S$
\STATE $\ket{\psi_{s}}_t=\expf{-iH_St}\ket{z_s}$
\STATE $(1,\ket{\psi_{s}^{H_S}}_t)\in \mathcal{S}$
\STATE  $\ket{\psi_{s+1}}_t=\expf{-iH_Gt}\ket{z_s}$
\STATE $(0,\ket{\psi_{s+1}^{H_G}}_t)\in \mathcal{S}$
\STATE s=s+2
\ENDWHILE
\end{algorithmic}
\end{algorithm}

We consider the case when $U(\vec{\theta})$ is a parametrized shallow HEA, and is $O$ a local operator measured at the output of the circuit. We define
\be
L_s(H,\vec{\theta},t)\coloneqq \tr[\ket{\psi_{s}^{H}}_t\!\!\bra{\psi_{s}^{H}}O(\vec{\theta})]\,,
\label{expvalue-costfunctionLs}
\ee
where $\ket{\psi_{s}^{H}}_t\in(y_s,\ket{\psi_{s}^{H}}_t)\in\mathcal{S}$ for $H\in\{H_G,H_S\}$, and $O(\vec{\theta})\coloneqq U^{\dag}(\vec{\theta})OU(\vec{\theta})$. Here, $y_s=1(0)$ if $H=H_S(H_G)$. In the following, we will drop the superscript in $\ket{\psi_s}\in\mathcal{S}$ to light the notation, unless necessary. Then, the goal is to minimize the empirical loss function:
\be
L(\vec{\theta},t)=\frac{1}{N}\!\sum_{(y_s,\ket{\psi_s})\in\mathcal{S}}\!\left(y_{s}-L_s(H,\vec{\theta},t)\right)^2\ ,
\ee
where $N$ is the size of the dataset $\mathcal{S}$. There are two necessary conditions for the success of the algorithm: $(i)$ the parameter landscape is not exponentially concentrated around its trivial value, and  $(ii)$ there exists $\vec{\theta}_{0}$ such that the model outputs are different for data in distinct classes. For instance, this can be achieved if $U^{\dag}(\vec{\theta}_0)OU(\vec{\theta}_0)=S$; as here  $L_{s}(H_S,\vec{\theta}_0,t)=1$ for any $s$ such that $\ket{\psi_s}\equiv \ket{\psi^{H_S}_{s}}_t$. Then, one also needs to have  $L_{s}(H_G,\vec{\theta}_0,t)$ not being close to one with high probability.  Note that if the symmetry $S$ is a local operator, and $O$ is chosen to be local, there are cases in which a shallow-depth HEA can find the solution  $U^{\dag}(\vec{\theta}_0)OU(\vec{\theta}_0)=S$. Such an example is shown below.

\subsection{Gaussian Diagonal Ensemble Hamiltonian discrimination}

Let us now specialize the example to an analytically tractable problem. We first show how the growth of the evolution time $t$, and thus the entanglement generation, affects the HEA's ability to solve the task. Then,  we show that there exists a critical time $t^{*}$ for which the states in the dataset satisfy an area law, and thus for which the QML algorithm can succeed. Since classically simulating random Hamiltonian evolution is a difficult task, the latter constitutes an example where a QML algorithm can enjoy a quantum speed-up with respect to classical machine learning.

Let $H_G$ be a random Hamiltonian, i.e., $H_G=\sum_{i}E_{i}\Pi_i$, where $\Pi_{i}$ are projectors onto random Haar states, and $E_{i}$ are normally distributed around $0$ with standard deviation $1/2$, (see App. \ref{App: isospec} for additional details). This ensemble of random Hamiltonians is called Gaussian Diagonal Ensemble ($\gde$), and it is the simplest, non-trivial example where our results apply. In Fig.~\ref{Fig: gdeham} we explicitly show how the time evolution under such a Hamiltonian can be implemented in a quantum circuit. Generalizations to wider used ensembles, such as Gaussian Unitary Ensemble ($\gue$), Gaussian Symplectic Ensemble ($\operatorname{GSE}$), Gaussian Orthogonal Ensemble ($\goe$), or the Poisson Ensemble ($\poi$), will be straightforward. We refer the reader to Refs.~\cite{Oliviero2020random,leone2020isospectral} for more details on these techniques. 

Consider a bipartition of $n$ qubits, i.e., $A\cup B$ such that $|A|\ll |B|$. Let $H_{S}$ a Random Hamiltonian commuting with all the operators on a local subsystem $A$, i.e. $[H_S,P_A]=0$ for all $P_{A}$. We can choose $H_{S}$ as $H_{S}=\mathds{1}_A\otimes H_{B}$, with $H_{B}$ belonging to the $\gde$ in the subsystem $B$. Let $H_G$ be a random Hamiltonian belonging to the $\gde$ ensemble in the subsystem $A\cup B$. Since the Hamiltonian $H_S$ commutes with all the operators in $A$, we choose the symmetry $S$ to be $S\equiv P_{A}\otimes \mathds{1}_{B}$, i.e. a Pauli operator with local support on $A$. To build the data-set $\mathcal{S}$, we thus identify the vector space containing all the eigenvectors with eigenvalue $1$ of $P_A$, $\mathcal{V}_{P_A}=\operatorname{span}\{\ket{z}\in\mathbb{C}^{2\otimes n}\,|\, P_A\ket{z}=\ket{z}\}$ and follow Algorithm \ref{alg1}. Note that, with this choice, $\operatorname{dim}(\mathcal{V}_{P_A})=2^{|A|-1}$ and thus we take $|A|\ge 2$. 
The QML task is to distinguish states evolved in time by $H_{G}$ or by $H_S$. Let us choose $O$ a Pauli operator having support on a local subsystem. Then, the following proposition holds:
\begin{prop}\label{propgde1}
Let $L_{s}(H,\vec{\theta},t)$ be the expectation value defined in Eq.~\eqref{expvalue-costfunctionLs}, for $H\in\{H_G,H_S\}$.  If there exist $\vec{\theta}_0$ such that $U^{\dag}(\vec{\theta}_0)OU(\vec{\theta}_0)=S$, then
\begin{align}
L_s(H_{S},\vec{\theta}_{0},t)&=1, \quad \forall t\,,\label{symmtheta0}\\
\mathbb{E}_{\operatorname{GDE}}[L_s(H_G,\vec{\theta},t)]&=e^{-t^2/4}+\mathcal{O}(2^{-n}), \quad \!\! \forall \vec{\theta}\,.\label{averagegde}
\end{align}
\end{prop}
See App.~\ref{App: isospec} for the proof. 

Notably, the symmetry of $H_{S}$ ensures that if the HEA is able to find $\vec{\theta}_0$, then  the output  $L_{s}(H_S,\vec{\theta}_0,t)=1$ is distinguishable from the expected value of $L_{s}(H_G,\vec{\theta}_0,t)$ which is exponentially suppressed in $t$. While in principle it is possible to minimize the loss function $L(\vec{\theta},t)$ for any $t$, the following theorem states that as the time $t$ grows, the parameter landscape get more and more concentrated, according to Definition~\ref{def: deterministic}.
\begin{theorem}[Concentration of loss for GDE Hamiltonians]\label{gdehamtheorem}
Let $L_{s}(H,\vec{\theta},t)$ be the expectation value defined in Eq.~\eqref{expvalue-costfunctionLs} for $H\in\{H_G,H_S\}$. For random  $\gde$ Hamiltonians one has
\be
\operatorname{Pr}\left(|L_s(H,\vec{\theta},s)|\le \epsilon\right)\ge 1-\epsilon^{-1} \,e^{-t^2/4}+\mathcal{O}(2^{-n})\,.
\ee
\end{theorem}
See App.~\ref{App: isospec} for the proof.

Note the above concentration bound holds for both $H_{S}$ and $H_{G}$, provided that $|A|\ll|B|$. From the above, one can readily derive the following corollary:
\begin{corollary}\label{gdehamcor}
Let $L_{s}(H,\vec{\theta},t)$ be the expectation value defined in Eq.~\eqref{expvalue-costfunctionLs} for $H\in\{H_G,H_S\}$. Then for $t\ge4\alpha\sqrt{n}/\log_2 e$, $\alpha>0$, and $\epsilon=e^{-\beta n}$, then:
\be
\operatorname{Pr}\left(|L_s(H,\vec{\theta},s)-L_{trv}|\le 2^{-\beta n}\right)\ge 1-2^{-(\alpha-\beta)n}
\ee
for any $\beta<\alpha$.
See App.~\ref{App: isospec}
\end{corollary}
Taken together, Theorem~\ref{gdehamtheorem} and Corollary~\ref{gdehamcor}, provide a no-go theorem for the success of the QML task, as they indicate that beyond $t\sim \sqrt{n}$ one encounters deterministic concentration with overwhelming probability. Crucially, the role of the entanglement generated by $H\in\{H_G,H_S\}$ is hidden in the variable $t$ of the bound in Theorem~\ref{gdehamtheorem}. Indeed, as shown in Refs.~\cite{Oliviero2020random,leone2020isospectral} the entanglement for random  $\gde$ Hamiltonians is monotonically growing with $t$.

\begin{figure}[t]
    \centering
    \includegraphics[width=\linewidth]{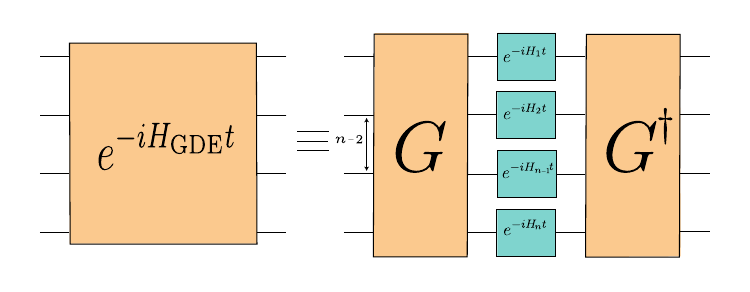}
    \caption{Circuit representation of the exponentiation of a $\gde$ Hamiltonian. The light-blue gates $e^{-iH_{k}t}$, are single qubit gates, exponentiation of $H_{k}\equiv \epsilon_{k}\st{0_k}+\epsilon_{k+1}\st{1_{k}}$. Labeling $i_{k}\in\{0_k,1_k\}$, the numbers $E_{\operatorname{dec}(\vec{i})}=\sum_{k=0}^{n-1}\epsilon_{k+i_k}$ are $2^n$ random numbers sampled from $\mathcal{N}(0,1/2)$, and corresponds to the eigenvalues of $H_{\gde}$; where $\operatorname{dec}(\vec{i})=\sum_{k}2^{k}i_{k}$. $G$ is a deep-random quantum circuit, building the eigenvectors of $H_{\gde}$.}
    \label{Fig: gdeham}
\end{figure}

While the previous results indicate that a HEA-based QML model will fail on the random Hamiltonian discrimination QML task for $t\sim \sqrt{n}$ (due to high levels of entanglement), this does not preclude the possibility of the model succeeding for smaller evolution times. Notably, here we can show that for $t=\mathcal{O}(\sqrt{\log (n)})$ the conditions are ideal for a quantum advantage: the states in the dataset will satisfy an area law of entanglement, and since $\gde$ Hamiltonians are built out of a very deep random quantum circuit, their time evolution can be classically hard. 
In particular, the following theorem holds.
\begin{theorem}[Area law and short-time evolution under $\gde$ Hamiltonians]\label{propgde2}
Let $H_{G}$ be a  $\gde$ Hamiltonian, and let $\ket{\psi_{0}}$ be any factorized state over the bipartition $\Lambda\cup \bar{\Lambda}$, with $\ket{\psi_t}=\expf{-iHt}\ket{\psi_0}$ being its unitary evolution under $H_G$. Let $\Lambda$ be a subsystem, where $|\Lambda|=\mathcal{O}(\log (n))$. Then,
\be
\operatorname{Pr}\left[\mathcal{I}_{\Lambda}(\psi_t)\ge \frac{e^{-t^2/8}}{2}\sqrt{1-2^{-|\Lambda|}}\right]\ge 1-2^{-n/2}\,,
\ee
thus for $t=\mathcal{O}(\sqrt{\log (n)})$, the evolved state satisfies area law of the entanglement with overwhelming probability
\be
\operatorname{Pr}\left[\mathcal{I}_{\Lambda}(\psi_t)\in\Omega\left(\frac{1}{\poly(n)}\right)\right]\ge 1-2^{-n/2}\,.
\ee
\end{theorem}
See App.~\ref{App: isospec} for the proof. Moreover,  the following corollary holds
\begin{corollary}
Let $H_{G}$ be a  $\gde$ Hamiltonian, and let $\ket{\psi_{0}}$ a factorized state, with $\ket{\psi_t}=\expf{-iHt}\ket{\psi_0}$ being its unitary evolution under $H_G$. Let $\Lambda$ be a subsystem, where $|\Lambda|=\mathcal{O}(\log (n))$. Then  the probability that the loss function $L(\vec{\theta},t)$ anti-concentrating is overwhelming.
\begin{proof}
The corollary easily descends from Proposition~\ref{propgde2}, Theorem~\ref{theore:anticoncentration} and Proposition~\ref{eq:prop-anticoncentration}.
\end{proof}
\end{corollary}
As shown above, for $t\in\mathcal{O}(\sqrt{\log(n)})$, the states generated by the time evolution of $\gde$ Hamiltonians obey to area law of the entanglement with overwhelming probability. Thanks to Theorem \ref{theore:anticoncentration}, we also have that the loss function $L(\vec{\theta},t)$ anti-concentrates, giving strong evidence of the success of the Hamiltonian Discrimination QML task.

\section{Numerical simulations}

In this section, we present numerical results which further explore the connection between the entanglement in the input state, and the phenomenon of gradient concentration. In particular, we are interested in showing how the parameter landscape of a QML problem becomes more and more concentrated as the entanglement in the input state grows.

\begin{figure}[t]
    \centering
    \includegraphics[width=0.5\columnwidth]{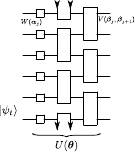}
\caption{Circuit of HEA used in numerical simulations. The architecture is composed of an initial layer of general single qubit unitaries. We denote these gates as $W(\vec{\alpha}_{j})\equiv e^{-i\vec{\sigma}_j\cdot \vec{\alpha}_j}$, and we parametrize them by $3$ angles $\vec{\alpha}_j\equiv(\alpha^{x}_{j},\alpha^{y}_{j},\alpha^{z}_{j})$. Here $\vec{\sigma}_j=(X_j,Y_j,Z_j)$. Then, two-qubit gates act on neighboring pairs of qubits. Each two-qubit gate denoted as $V_{j,j+1}(\vec{\beta}_j,\vec{\beta}_{j+1})$, is composed of a CNOT followed by a general single qubit gate on each qubit. That is, $V_{j,j+1}(\vec{\beta}_j,\vec{\beta}_{j+1})= W(\vec{\beta}_j)W(\vec{\beta}_{j+1})\cnot_{j,j+1}$. }
    \label{Fig: LfigHEA}
\end{figure}

To create $n$-qubit states with different amounts of entanglement, we will consider time-evolved states of the form  
\begin{equation}\label{eq:input_state}
    \ket{\psi_t}\coloneqq \expf{-iHt}\ket{\psi_0}\,,
\end{equation}
where $\ket{\psi_0}$ is a random product state, and where $H$ is the Heisenberg model with  first-neighbor interactions
\be
H=\sum_{i=1}^{n}X_iX_{i+1}+Y_{i}Y_{i+1}+2Z_{i}Z_{i+1}+X_i\, ,
\label{Hamiltoniannumerics}
\ee
with periodic boundary conditions ($n+1\equiv 1$). Here, $\sigma_i$ with $\sigma=X,Y,Z$, denotes a Pauli operator acting on qubit $i$. As we will see below, as $t$ increases, so does the entanglement in $\ket{\psi_t}$. 

Next, we will consider a learning task where we want to minimize a cost-function of the form
\be
L(\vec{\theta})=1-\tr[U(\vec{\theta})\st{\psi_t}U^{\dag}(\vec{\theta})O_Z] \,,
\label{costnumerics}
\ee
where $O_Z=\sum_{i}Z_i$ (i.e., $O_Z$ is a sum of $1$-local operators), and where $U(\vec{\theta})$ is a shallow HEA. Specifically, we employ the HEA architecture shown in Fig.~\ref{Fig: LfigHEA} which is composed of an initial layer of general single qubit rotations, followed by two-qubit gates on alternating pairs of qubits. The two-qubit gates are themselves composed of a CNOT gate followed by general single-qubit gates on each qubit.  

\begin{figure*}
    \centering
    \includegraphics[width=.8\textwidth]{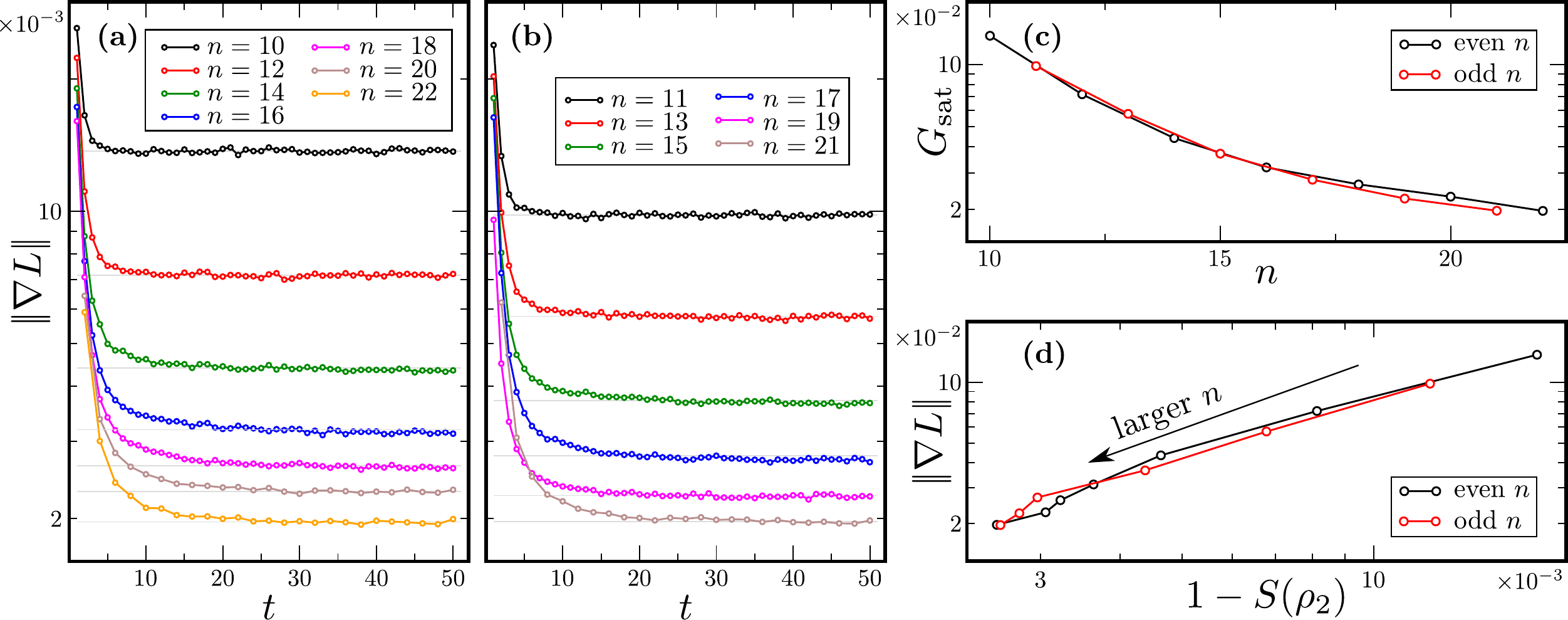}
    \caption{Numerical results. We consider a problem where the input states $\ket{\psi_t}$ of the HEA are determined by Eqs.~\eqref{eq:input_state} and~\eqref{Hamiltoniannumerics} and where  the loss function is given by~\eqref{costnumerics}. In panels (a) and (b) we respectively show (averaged over $\ket{\psi_t}$ and $\vec{\theta}$) norm of the gradient $\partial_\mu L(\vec{\theta})$ as a function of the evolution time $t$ for different system sizes  with $n$ even and odd. Panel (c) shows norm saturation value $G_{\text{sat}}$ of the results presented in (a) and (b) versus system size $n$. Panel (d) shows the norm of the gradient versus $1-S(\rho_2)$, where $S(\rho_2)$ is the entropy of two-qubit subsystem for an evolution time such that the saturation value $G_{\text{sat}}$ is achieved. Different points correspond to different values of $n$. The results are averaged over 400 initial states $\ket{\psi_0}$ in Eq.~\eqref{eq:input_state} and two sets of angles $\vec{\theta}$ for every initial state. }
    \label{Fig: numerics}
\end{figure*}

In Figs.~\ref{Fig: numerics}(a,b) we show averaged norm of the gradient $\partial_{\mu}L(\vec{\theta})$, i.e. $\norm{\nabla L}_\infty \equiv \max_\mu |\partial_{\mu}L(\vec{\theta})|$, as a function of the evolution time $t$ used to prepare the input state of the HEA for different problem sizes. Gradients are computed by averaging over $400$ random product states $\ket{\psi_0}$, and two sets of random parameters in the HEA for each initial state. Here we can see that for small evolution times  the cost exhibits large gradients independently of the system size. This result is expected as we recall that in the limit $t\rightarrow 0$ the input state $\ket{\psi_0}$ is a tensor product state, which, along with $1$-local measurements  and the HEA structure, leads to the gradients which norms are independent of $n$. As $t$ increases, we can see that the gradient norm decrease until a saturation value $G_{\text{sat}}$ is achieved. Moreover, we can see that the value of $G_{\text{sat}}$ depends on the number of qubits in the system. In fact, as shown in Fig.~\ref{Fig: numerics}(c), $G_{\text{sat}}$ decays polynomially with $n$. We can further understand this behavior by noting that as $t$ increases, the time-evolution $\expf{-iHt}$ produces larger amounts of entanglement in the input state, and concomitantly smaller gradients (as indicated by our main results above). To see that this is the case, we compute rescaled entropy $S(\rho_2)=-\Tr[\rho_2 \log(\rho_2)]/2$ where $\rho_2$ is the reduced state on two nearest-neighbor qubits for a sufficiently large time $t$ such that $G_{\text{sat}}$ is achieved. Results are shown in 
Fig. \ref{Fig: numerics}(c). It shows a positive correlation between the decay of gradients and the increase in reduced state entropy. Thus, the more entanglement in the input state, the smaller the gradients, and the more concentrated the landscape.

\section{Discussion and conclusions}

Understanding the capabilities and limitations of VQA and QML algorithms is crucial to developing strategies that can be used to achieve a quantum advantage. One of the most relevant ingredients in ensuring the success of a VQA/QML model is the choice of ansatzes for the parametrized quantum circuit. In this work, we focused our attention on the shallow HEA, as it can avoid barren plateaus, and since it is perhaps one of the most NISQ-friendly ansatzes. Currently, the HEA is widely used for a plethora of problems, irrespective of whether it is well-fit for the task and data at hand. In a sense, the HEA is still a ``\textit{solution in search of a problem}'' as there was no rigorous study of the tasks where it should, or should not be used. In this work, we establish rigorous results, showing how, and in which contexts, HEAs are (and are not) useful and can eventually provide a signature of quantum advantage.

We first review relevant results from the literature, discussing the notion of cost and loss function concentration and necessary conditions for trainability of HEAs -- i.e. shallowness and locality of measurements. Here we highlight the existence of a new source of untrainability of shallow HEAs:  the entanglement of the input states. On one hand, we proved that HEAs are untrainable if the input states satisfy a volume law of entanglement, as the cost function is deterministically concentrated around its trivial value. On the other hand, if the input states follow an area law of entanglement, the HEA is trainable. In fact, here we prove that the loss function  anti-concentrates, i.e., it differs, at least polynomially, at sufficiently different points of the parameters landscape.

While the role of entanglement in the trainability of VQA and QML models has  been explored in Refs.~\cite{marrero2020entanglement,patti2020entanglement}, the results found therein are conceptually different from ours. Namely, in these references the authors point out that deep parametrized quantum circuit ansatzes create volume law for the entanglement entropy, making the parameter landscape exponentially flat in the number of qubits and thus giving rise to entanglement-induced barren plateaus. As such, these results study the entanglement created \textit{during} the circuit, but not that \textit{already present} in the input states. For instance, the shallow HEA cannot create volume law of entanglement, yet, it is still untrainable if such entanglement exists in the input state. Hence, our work provides a new source of untrainability for certain datasets.

Next, we also analyzed the still open question of whether the HEA is able to achieve a quantum advantage in a VQA/QML setting. While the answer is far beyond the scope of the paper, we identified regimes in which the HEA can or cannot provide quantum speed-ups. Here we  proved that thanks to the shallowness of HEAs, input states with bond dimensions at most polynomially in the number of qubits can be simulated with only a polynomial overhead on classical machines. This result rules out the use of HEA in VQAs: as many examples show, the typical input state for a VQA is an easy-to-prepare product state, thus allowing an efficient classical decomposition. Conversely, for QML algorithms the question still remains open: the portion of area law states admitting an efficient classical description is exponentially small~\cite{ge2016area}. While this is not a guarantee for achieving quantum advantage, this is definitely the window to look at for  applications beyond those solvable by classical capabilities. 

We indeed push forward the latter intuition and provide an example to which our results apply. Namely,  we present a Hamiltonian discrimination QML problem, where initial product states are evolved in time by two types of Hamiltonians, one possessing a given local symmetry, and one completely general. We show that, while the task becomes less and less feasible if the evolution time is long (as entanglement growing in time), for a given time window (scaling logarithmically with the number of qubits) such states possess area law of entanglement, ensuring the absence of barren plateaus in the loss landscape.

Such an example serves as a pivotal one for future, and hopefully fruitful, usages of the HEA. Our recipe to prepare a Barren-plateau-free QML problem is the following: consider entangled enough quantum input data, pass it through a shallow-depth circuit, and then measure with local operators. Importantly, if one wishes to use this scheme for classical data, it will be extremely important to find data-embedding schemes that lead to area-law states, but which are not themselves classically simulable. We expect that the search for such entangled-tamed embedding could be a fruitful research direction.

There are still several directions to be explored after the analysis of the present work. We indeed emphasize that, while for unstructured HEAs, this paper definitely rules out volume law states as input states of QML algorithms with HEA ansatzes, there is the fascinating possibility that, with even a little prior knowledge of the input states, a problem-aware HEA  could avoid exponential concentration in the parameter landscapes. Indeed, the choice of some structured, and problem-dependent ansatz can avoid barren plateaus: a prominent example is the Geometric Quantum Machine Learning, which exploits the geometric symmetries of the input data-set to design symmetry-aware parametrized ansatzes~\cite{larocca2022group,meyer2022exploiting,skolik2022equivariant,sauvage2022building,
nguyen2022atheory,schatzki2022theoretical,ragone2022representation}.

\acknowledgements

This work was supported by the U.S.Department of Energy (DOE) through a quantum computing program sponsored by the Los Alamos National Laboratory Information Science \& Technology Institute.  L.L and S.F.E.O. acknowledge support from NSF award no. 2014000 and by the  Center for Nonlinear Studies at Los Alamos National Laboratory (LANL). L.C. was partially supported by the U.S. DOE, Office of Science, Office of Advanced Scientific Computing Research, under the Accelerated Research in Quantum Computing (ARQC) program. M.C. was initially supported by the Laboratory Directed Research and Development (LDRD) program of LANL under project number 20230049DR.
The authors also acknowledge support by the U.S. DOE through a quantum computing program sponsored by the LANL Information Science \& Technology Institute.

\bibliographystyle{quantum}
\bibliography{VQAent}

\appendix
\onecolumngrid
\section{Proof of Theorem~\ref{th1} and Corollaries~\ref{cor2}, \ref{application2}, and~\ref{app2copies}}\label{App: proofth1}
\subsection{Proof of Theorem~\ref{th1}}
Let us consider the following distance $|f(\vec{\theta})-f_{trv}|$, where $f(\vec{\theta})=C(\vec{\theta})$ and $f_{trv}=2^{-n}\tr[O]$. Our task is to upper-bound this distance with the information-theoretic measure of information scrambling $\mathcal{I}_{\Lambda}(\psi)$.
Let us decompose the measurement operator $O$ on the Pauli basis:
\be
O=\sum_{i=1}^{N}c_{i}P_{i}\,,
\label{operatordef}
\ee
where $c_{i}\coloneqq \tr[OP_i]/2^n$, and $\tr[P_iP_j]=2^n\delta_{ij}$. Then, we define the support $\supp(P_i)$ as the ordered set of qubits containing non-identity operators in the tensor product structure of each Pauli operator $P_i=\bigotimes_{j=1}^{n}\sigma_{j}^{(i)}$, where $\sigma_{j}^{(i)}$ is a single qubit Pauli operator acting on the $j$-th qubit. Explicitly, $\supp(P_i)=\{q_{1}^{(i)},\ldots, q_{S_i}^{(i)}\,|\, q_{1}^{(i)}< \ldots<q_{S_i}^{(i)}, q_{k}^{(i)}\in [1,n]\}$ is an ordered subset of natural numbers labeling the qubits on which $P_i$ acts non-trivially, and $S_i$ labels the number of qubits on which $P_i$ acts, see Fig.~\ref{fig:cluster}. 
\begin{figure}[t]
    \centering
    \includegraphics[width=1\linewidth]{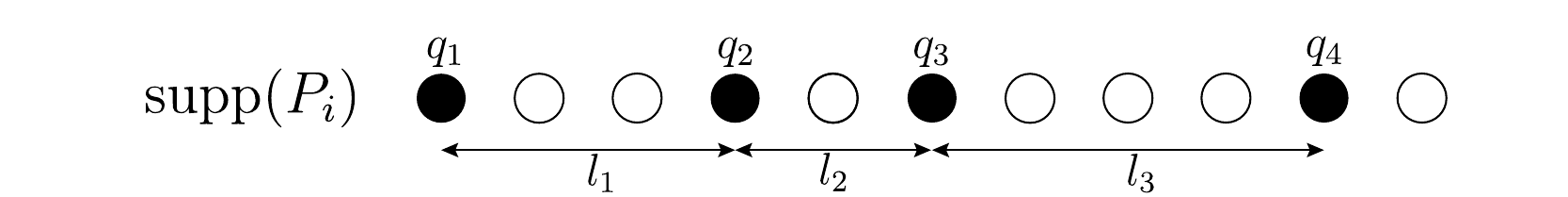}
    \caption{ The support $\supp(P_{i})$ contains $q_{1},\ldots, q_{4}$ qubits, whose relative distances are $l_{1},\ldots, _{3}$ qubits, used to define clusters when compared to $2D$.}
    \label{fig:cluster}
\end{figure}
From this, we also define the support of $O$ as the ordered subset of qubits given by the union of the supports of each $P_i$, i.e., $\supp(O)=\bigcup_{i}\supp(P_i)$.  Moreover, we note that $\sum_{i}|c_i|\le N\max_{i}|c_i|\le N\norm{O}_{\infty}$ (which the latter follows from Holder's inequality: $|c_{i}|= 2^{-n}|\tr(OP_i)|\le 2^{-n}\|O\|_{\infty}\|P_i\|_1=\|O\|_{\infty}$), and $N\le (3n)^{\max_i|\supp(P_i)|}$ (which follows from a counting argument). Let us define the pairwise relative distance between the qubits $q_{k}$ and $q_{k+1}$ belonging to $\supp(P_{i})$ as:
\be
l_{k}^{(i)}\coloneqq |q_{k+1}^{(i)}-q_{k}^{(i)}|, \quad k=1,\ldots, S_{i}-1\,.
\ee
It is now possible to define clusters, i.e., subsets of contiguous qubits whose pairwise relative distance is less than $2D$. The definition can be done recursively: let $C_{1}^{(i)}$ be the first cluster, then $q_{1}^{(i)}\in C_{1}^{(i)}$; if $l_{1}^{(i)}\le 2D$ then $q_{2}^{(i)}\in C_{1}^{(i)}$, otherwise $q_{2}^{(i)}\in C_{2}^{(i)}$. This procedure defines $C_{1}^{(i)},\ldots, C_{L_i}^{(i)}$ clusters of qubits for any $i=1,\ldots,N$, such that $1\le L_{i}\le S_i$. Note that, we can write the operator $O$ as
$O=\sum_{i}c_i\bigotimes_{\alpha=1}^{L_i}O_{\alpha}^{(i)}$, with $\supp(O_{\alpha}^{(i)})= C_{\alpha}^{(i)}$.
Consequently, we can rewrite the cost function by expanding $O=\sum_i c_i P_i=\sum_{i}c_i(\bigotimes_{\alpha}O_{\alpha}^{(i)})$:
\be
f(\vec{\theta})=\sum_{i}c_{i}\tr[\st{\psi}U^{\dag}(\vec{\theta})(\bigotimes_{\alpha}O_{\alpha}^{(i)})U(\vec{\theta})]\,.
\ee
It is possible to evaluate the distance between $f(\vec{\theta})$ and $f_{trv}$, that can be rewritten as:
\be
f_{trv}= 2^{-n}\sum_{i}c_{i}\tr[U^{\dag}(\vec{\theta})(\bigotimes_{\alpha}O_{\alpha}^{(i)})U(\vec{\theta})]\,.
\ee 
Defining $\Lambda_i\coloneqq \supp(U^{\dag}(\vec{\theta})(\bigotimes_{\alpha}O_{\alpha}^{(i)})U(\vec{\theta}))$, the distance between $f(\vec{\theta})$ and $f_{trv}$ reads:
\begin{equation}
\begin{aligned}
|f(\vec{\theta})-f_{trv}|&\le \sum_{i}|c_i||\tr_{\Lambda_i}[\psi_{\Lambda_i}U^{\dag}(\vec{\theta})(\bigotimes_{\alpha}O_{\alpha}^{(i)})U(\vec{\theta})]-2^{-|\Lambda_i|}\tr_{\Lambda_i}[U^{\dag}(\vec{\theta})(\bigotimes_{\alpha}O_{\alpha}^{(i)})U(\vec{\theta})]|\\
&\le\sum_{i}|c_{i}|\norm{O_{\alpha}^{(i)}}_{\infty}\norm{\psi_{\Lambda_i}-\frac{\mathds{1}_{\Lambda_i}}{2^{|\Lambda_i|}}}_{1}\equiv\sum_{i}|c_{i}|\mathcal{I}_{\Lambda_i}(\psi)\,,
\end{aligned}
\end{equation}
where $\psi_{\Lambda}\equiv \tr_{\bar{\Lambda}}\st{\psi}$. In the first inequality, we use the triangle inequality of absolute value, while in the second one, we used $|\tr[AB]|\le \norm{A}_{\infty}\norm{B}_{1}$, and the fact that unitary operators preserve the trace and any Schatten $p$-norm. Finally, we used  $|\!|\mathcal{O}_{\alpha}^{(i)}|\!|_{\infty}=1$. To complete the proof it is still necessary to bound $|\Lambda|\equiv \max_{i}|\Lambda_i|$.

In what follows we will use the following lemma and corollary, which we will prove at the end of this Appendix (see App.~\ref{App:lem}).
\begin{lemma}\label{lemma1}
Let $U(\vec{\theta})$ be a HEA with depth $D$ as in Fig.~\ref{fig:HEA}(b), and $O=\sum_{i}c_i P_i=\sum_{i}c_i\bigotimes_{\alpha}O_{\alpha}^{(i)}$ be an operator, as in Eq.~\eqref{operatordef}. Then for any $O_{\alpha}^{(i)}$:
\be
|\supp(U(\vec\theta)O_{\alpha}^{(i)}U^{\dag}(\vec\theta))|\le \sum_{\substack{k\in C_{\alpha}^{(i)}\\ l_{k}^{(i)}< 2D}} l_{k}^{(i)}+2D\,.
\ee
\end{lemma}
The following corollary descends from the above lemma and the clustering of qubits. 
\begin{corollary}\label{corollary1}
Let $U(\vec{\theta})$ be HEA with depth $D$, with $D\in \mathcal{O}(\log (n))$ and $O=\sum_{i}c_{i}P_{i}=\sum_{i}c_i\bigotimes_{\alpha}O_{\alpha}^{(i)}$. Then for any $O_{\alpha}^{(i)},O_{\alpha^\prime}^{(i)}$, with $\alpha\neq \alpha^{\prime}$, one finds that asymptotically (in $n$), 
\be
\supp(U(\vec\theta)O_{\alpha}^{(i)}U^{\dag}(\vec\theta))\cap \supp(U(\vec\theta)O_{\alpha^\prime}^{(i)}U^{\dag}(\vec\theta))=\emptyset\,.
\ee
\end{corollary}

 Thanks to Corollary~\ref{corollary1}, we know that $\supp(U^{\dag}(\vec\theta)O_{\alpha}U(\vec\theta))\cap \supp(U^{\dag}(\vec\theta)O_{\alpha^{\prime}}U(\vec\theta))=\emptyset$, and thus:
\be
|\Lambda|=\max_i\sum_{\alpha}|\supp(U^{\dag}(\vec\theta)O_{\alpha}^{(i)}U(\vec\theta))|
\ee
then, from Lemma~\ref{lemma1}:
\begin{align}
|\Lambda|&\le\max_i \sum_{\alpha}\sum_{\substack{k\in C_{\alpha}^{(i)}\\ l_{k}<mD}}l_{k}^{(i)}+mD \le\max_i \sum_{\alpha}(|C_{\alpha}^{(i)}|-1)mD+mD\,,
\end{align}
where the second inequality follows by assuming maximum distance between qubits in the cluster. Clearly $\sum_{\alpha}|C_{\alpha}^{(i)}|=|\supp(P_{i})|$ by definition, and $\sum_{\alpha} 1=L_i\ge 1$, thus
 $mD(1-\sum_{\alpha} 1)\le 0$, and we obtain the following
 \be
 |\Lambda|\le mD\max_{i}|\supp(P_{i})|\,.
 \ee

\subsection{Proof of Corollary~\ref{cor2}}
By theorem~\ref{th1}, we have
\be
|f(\vec{\theta})-f_{trv}|\le\norm{O}_{\infty} (3n)^{\max_{i}\supp(P_i)}\max_{i}\mathcal{I}_{\Lambda_i}(\psi)\,,
\ee
if the information is scrambled, according to Definition~\ref{def-scrambling}, we have that, since $\max_{i}|\Lambda_i|\in\mathcal{O}(\log (n))$, then $\max_{i}\mathcal{I}_{\Lambda_i}(\psi)\in\mathcal{O}(2^{-cn})$ for some $c>0$. Then, if $\max_{i}|\supp(P_i)|\in\mathcal{O}(\log (n))$,  there exists a constant $c^{\prime}$ such that $n^{\max_i\supp(P_i)}\in \mathcal{O} (n^{c^{\prime}\log (n)})$, and therefore:
\be
(3n)^{\max_i\supp(P_i)}\in\mathcal{O} (n^{2c^{\prime}\log (n)})\, ,
\ee
where we crudely upperbounded $3n<n^2$ (for $n>3$). Now, we have that $n^{2c^{\prime}\log (n)}=2^{2c^{\prime}\log^2n}\in\mathcal{O}(2^{kn})$ for any $k\in\Omega(1)$. Choosing $k=c/2$ for simplicity, we finally reach the final bound:
\be
|f(\vec{\theta})-f_{trv}|\in\mathcal{O}(2^{-cn/2})\,.
\ee
where we considered $\norm{O}_{\infty}\in\Omega(1)$.

\subsection{Proof of Corollary~\ref{application2}}
The proof of Corollary~\ref{application2} follows from the proof in the work of Popescu et al.~\cite{popescu2006entanglement}. Here we report it for completeness.

Let $\psi_{\Lambda}=\tr_{\bar\Lambda}[\st{\psi}]$, consider a complete set of (Hermitian) observables $O_{\Lambda}^{(i)}$ on $\Lambda$ with $\norm{O_{\Lambda}^{(i)}}_\infty=1$, and decompose the state on the complete set of observables $O_{\Lambda}^{(i)}$:
\be
\psi_{\Lambda}=\frac{1}{d_{\Lambda}}\sum_{i}\tr[\psi_{\Lambda}O_{\Lambda}^{(i)}]O_{\Lambda}^{(i)}\,.\label{eq:psi_lambda_expansion}
\ee
Consider the $1$-norm distance between $\psi_{\Lambda}$ and the completely mixed state $\mathds{1}_{\Lambda}d_{\Lambda}^{-1}$:
\begin{equation}
\begin{aligned}
\mathcal{I}_{\Lambda}(\psi)&\le \sqrt{d_{\Lambda}}\norm{\psi_{\Lambda}-\frac{\mathds{1}_{\Lambda}}{d_{\Lambda}}}_2\le\sqrt{\sum_{i}\left(\tr\left[\psi_{\Lambda}O_{\Lambda}^{(i)}\right]-\frac{\tr\left[O_{\Lambda}^{(i)}\right]}{d_{\Lambda}}\right)^2}\\
&\le d_{\Lambda}\max_{i}\left|\tr[\psi_{\Lambda}O^{(i)}_{\Lambda}]-\frac{\tr\left[O_{\Lambda}^{(i)}\right]}{d_{\Lambda}}\right|\,.
\end{aligned}
\end{equation}
In the first inequality, we have used the norm equivalence, while in the second we have expanded $\psi_{\Lambda}$ as in Eq.~\eqref{eq:psi_lambda_expansion}. The third inequality follows by simply taking the maximum over $i$ in the summation.

Thus:
\begin{align}
&\operatorname{Pr}\left(\norm{\psi_{\Lambda}(\psi)-\mathds{1}_{\Lambda}d_{\Lambda}^{-1}}_1\ge d_{\Lambda}\epsilon\right)\le\operatorname{Pr}\left(\max_{i}\left|\tr\left[\psi_{\Lambda}O^{(i)}_{\Lambda}\right]-\frac{\tr\left[O_{\Lambda}^{(i)}\right]}{d_{\Lambda}}\right|\ge \epsilon\right)\,.
\end{align}
Since the expectation values are Lipschitz functions~\cite{popescu2006entanglement}, by exploiting Levy's lemma  one can easily prove that:
\be
\operatorname{Pr}\!\left(\!\max_{i}\left|\tr\left[\psi_{\Lambda}O^{(i)}_{\Lambda}\right]-\frac{\tr\left[O_{\Lambda}^{(i)}\right]}{d_{\Lambda}}\right|\ge \epsilon\!\right)\!\le 2d_{\Lambda}^2e^{-\frac{C2^n\epsilon^2}{2}}\nonumber\,,
\ee
where $C=(18\pi^3)^{-1}$. By choosing $\epsilon=2^{-1/3n}$, one gets:
\be
\operatorname{Pr}(\mathcal{I}_{\Lambda}<2^{|\Lambda|-1/3n})\ge1-2e^{-C2^{n/3-1}+2|\Lambda|\ln 2}\,.
\ee
To conclude, in virtue of Theorem~\ref{th1}, we have:
\be
|f(\vec{\theta})-f_{trv}|\le\norm{O}_{\infty} n^{\max_{i}\supp(P_i)}\max_{i}\mathcal{I}_{\Lambda_i}(\psi)\,.
\ee
Thanks to the above equation, we can bound $\mathcal{I}_{\Lambda_i}<2^{|\Lambda_i|}2^{-n/3}$ with overwhelming probability. Thus, $\max_{i}\mathcal{I}_{\Lambda_i}<2^{|\Lambda|}2^{-n/3}$, for $\Lambda$ being such that $|\Lambda|=\max_{i}|\Lambda_i|$. Thus, we have $|f(\vec{\theta})-f_{trv}|\le\norm{O}_{\infty} (3n)^{\max_{i}\supp(P_i)}2^{\max_{i}\supp(P_i)+mD}2^{-n/3}$. Now, given $\max_{i}\supp(P_i)\in\mathcal{O}(\log (n))$, and $D\in\mathcal{O}(\log (n))$, it means that there exists two constants $c,c^{\prime}$ such that 
\be
|f(\vec{\theta})-f_{trv}|\le\norm{O}_{\infty} n^{c+c^{\prime}\log (n)}2^{-n/3}\,.
\ee
To conclude, we can loosely write that for any constant $k$ we have $n^{c+c^{\prime}\log (n)}\in\mathcal{O}(2^{n/k})$. Choosing $k=6$ for simplicity, we finally reach the final bound:
\be
|f(\vec{\theta})-f_{trv}|\in\mathcal{O}(2^{-n/6})\,,
\ee
holding with probability 
\be
\operatorname{Pr}(|f(\vec{\theta})-f_{trv}|\in\mathcal{O}(2^{-n/6}))\ge 1-e^{-C2^{n/3}+(2n+1)\ln 2}\,,
\ee
where we bound $|\Lambda|<n$.
\subsection{Proof of Corollary~\ref{app2copies}}
Now suppose one has as input $k$ copies of a Haar random state $\ket{\psi}$ on $n$ qubits, denote it as:
\be
\ket{\Psi^{(k)}}=\ket{\psi}^{\otimes k}\,,
\ee
and denote $d=2^{n}$ the dimension of the Hilbert space in which $\ket{\psi}$ is living. Let us decompose its reduced density matrix $\Psi^{(k)}_{\Lambda}\coloneqq \tr_{\bar{\Lambda}}\st{\Psi^{(k)}}$ in a Hermitian operator basis $P_i$:
\be
\Psi^{(k)}_{\Lambda}=\frac{1}{d_{\Lambda}}\sum_{i}\tr[\Psi^{(k)}_{\Lambda}P_{i}]P_{i}\,.
\ee
Let us prove that each expectation value on $\ket{\Psi^{(k)}}$ is a Lipschitz function with respect to $\ket{\psi}$. Given a second state $\tilde{\Psi}^{(k)}_{\Lambda}$, we have
\begin{align}
|\tr[\Psi^{(k)}_{\Lambda}P_{i}]-\tr[\tilde{\Psi}^{(k)}_{\Lambda}P_{i}]|&\le\norm{\st{\psi}^{\otimes k}-\st{\tilde{\psi}}^{\otimes k}}_1 \nonumber\\
&\le k\norm{\st{\psi}-\st{\tilde{\psi}}}_{1}\nonumber\\
&\le2k\sqrt{1-|\braket{\psi|\tilde{\psi}}|^2}\nonumber\\
&\le2k \norm{\ket{\psi}-\ket{\tilde{\psi}}}\,,
\end{align}
in the first inequality we used the fact that $\norm{P_{i}}_{\infty}=1$. In the second line, we used $k$ times the following trick. Denote $\psi=\st{\psi}$:
\begin{equation}
\begin{aligned}
\norm{\psi^{\otimes k}-\tilde{\psi}^{\otimes k}}_1&=\norm{\psi^{\otimes k}-\psi\otimes\tilde{\psi}^{\otimes k-1}+\psi\otimes\tilde{\psi}^{\otimes k-1}-\tilde{\psi}^{\otimes k}}_1\\
&\le \norm{\psi^{\otimes k-1}-\tilde{\psi}^{\otimes k-1}}_1+\norm{\psi-\tilde{\psi}}_1\,,
\end{aligned}
\end{equation}
and we used $\norm{\psi^{\otimes k}}_{1}=1$ for any $k$. Thanks to the fact that $\tr[\Psi_{\Lambda}^{(k)}P_{i}]$ is Lipschitz, and denoting as $\overline{\tr\left[\Psi_{\Lambda}^{(k)}P_{i}\right]}$ the Haar average over the input state,  we know that:
\be
\operatorname{Pr}\left(\left|\tr\left[\Psi_{\Lambda}^{(k)}P_{i}\right]-\overline{\tr\left[\Psi_{\Lambda}^{(k)}P_{i}\right]}\right|\ge \epsilon\right)\le e^{-C2^{n}\epsilon^2/2k^2}\,,
\ee
where $C=(18\pi^3)^{-1}$. We need to bound the probability that $\mathcal{I}_{\Lambda}(\psi)\ge \epsilon$. Using the same trick as in Corollary~\ref{application2}, we can write:
\begin{equation}
\operatorname{Pr}\left(\mathcal{I}_{\Lambda}(\Psi^{(k)})\ge d_{\Lambda}\epsilon\right)\le\operatorname{Pr}\left(\max_{i}\left|\tr\left[\Psi^{(k)}_{\Lambda}O^{(i)}_{\Lambda}\right]-\frac{\tr\left[O_{\Lambda}^{(i)}\right]}{d_{\Lambda}}\right|\ge \epsilon\right)\nonumber\,.
\end{equation}
 Note that the following inequality holds
 \begin{equation}
\begin{aligned}
\left|\tr\left[\Psi^{(k)}_{\Lambda}O^{(i)}_{\Lambda}\right]-\frac{\tr\left[O_{\Lambda}^{(i)}\right]}{d_{\Lambda}}\right|\le  
\left|\tr\left[\Psi^{(k)}_{\Lambda}O^{(i)}_{\Lambda}\right]-\overline{\tr\left[\Psi^{(k)}_{\Lambda}O^{(i)}_{\Lambda}\right]}\right|+\left|\overline{\tr\left[\Psi^{(k)}_{\Lambda}O^{(i)}_{\Lambda}\right]}-\frac{\tr\left[O^{(i)}_{\Lambda}\right]}{d_{\Lambda}}\right|\,,
\end{aligned}
\end{equation}
then
\begin{equation}
\begin{aligned}
\left|\overline{\tr\left[\Psi^{(k)}_{\Lambda}O^{(i)}_{\Lambda}\right]}-\frac{\tr\left[O^{(i)}_{\Lambda}\right]}{d_{\Lambda}}\right|&=\frac{\left|\tr\left[\tr_{\bar{\Lambda}}\Pi_{\sym}O^{(i)}_{\Lambda}\right]-\frac{\tr\left[O^{(i)}_{\Lambda}\right]}{d_{\Lambda}}\right|}{D_{\sym}}\nonumber\\&\le \norm{\frac{\tr_{\bar{\Lambda}}[\Pi_{\sym}^{(k)}]}{D_{\sym}^{(k)}}-\frac{\mathds{1}_{\Lambda}}{d_{\Lambda}}}\,,
\end{aligned}
\end{equation}
where the last inequality follows from Holder's inequality. 
Here we have also used the known relation~\cite{harrow2013church}
\be
\frac{\Pi_{\sym}^{(k)}}{D^{(k)}_{\sym}}=\int\de\psi\st{\psi}^{\otimes k}=\frac{1}{D_{\sym}^{(k)}}\sum_{\pi}T_{\pi}\,,
\ee
where $T_{\pi}$ are permutation operators between the copies of $\ket{\psi}$. Let us compute $\tr_{\bar{\Lambda}}[\Pi_{\sym}^{(k)}]/D_{\sym}^{(k)}$, for the simplified case of $k=2$. For $k=2$ we have $\Pi_{\sym}^{(k)}/D_{\sym}^{(k)}=\frac{\mathds{1}+T}{d(d+1)}$, where $T$ denotes the SWAP operator. Then:
\be
\tr_{\bar{\Lambda}}\left[\frac{\mathds{1}+T}{d(d+1)}\right]=\frac{d_{\bar{\Lambda}}}{d(d+1)}\mathds{1}_\Lambda+\frac{\tr_{\bar{\Lambda}}[T]}{d(d+1)}\,.
\ee
To compute the partial trace, let us define $\Lambda_{12}\coloneqq \{(i,n+i)\,|\, i\in\Lambda,\, n+i\in\Lambda\}$, and $\Lambda_{12}^{c}\coloneqq \{(i,n+i)\,|\, i\in\bar{\Lambda},\, n+i\in\bar{\Lambda}\}$, then $\Lambda_{1}\coloneqq \{(i,n+i)\,|\, i\in\Lambda,\, n+i\in\bar{\Lambda}\}$, and $\Lambda_{2}\coloneqq \{(i,n+i)\,|\, i\in\bar{\Lambda},\, n+i\in{\Lambda}\}$. Note that $\Lambda_{12}\cap\Lambda_{12}^{c}=\emptyset$,$\Lambda_{12}\cap\Lambda_{1}=\emptyset$ etc., but $\Lambda_{12}\cup\Lambda_{12}^{c}\cup\Lambda_{1}\cup\Lambda_2=2n$. We can thus write the swap $T$ as:
\be
T=T_{\Lambda_{12}}T_{\Lambda_{12}^{c}}T_{\Lambda_{1}}T_{\Lambda_2}\,,
\ee
the partial trace thus can be written as:
\be
\tr_{\bar{\Lambda}}[T]=T_{\Lambda_{12}}\tr_{12}[T_{\Lambda_{12}^{c}}]\tr_{2}[T_{\Lambda_1}]\tr_{1}[T_{\Lambda_2}]\,,
\ee
where the labels mean, trace over all the pair of qubits $\tr_{12}$, trace over the first qubits $\tr_1$, and trace over the second qubits $\tr_{2}$. We have $\tr_{12}[T_{\Lambda_{12}^{c}}]=2^{|\Lambda_{12}^{c}|}<d_{\bar{\Lambda}}$, $\tr_{1}[T_{\Lambda_2}]=\mathds{1}_{2}$, and $\tr_{2}[T_{\Lambda_1}]=\mathds{1}_1$. At the end, we have:
\begin{equation}
\begin{aligned}
\norm{\tr_{\bar{\Lambda}}[\frac{\mathds{1}+T}{d(d+1)}]-\frac{\mathds{1}_{\Lambda}}{d_{\Lambda}}}_1\nonumber&=\left|\frac{d_{\bar{\Lambda}}d_{\Lambda}}{d(d+1)}-1\right|\norm{\frac{\mathds{1}_{\Lambda}}{d_{\Lambda}}}_1 +\norm{\frac{\tr_{\bar{\Lambda}}[T]}{d(d+1)}}_1\\
&\le\mathcal{O}(d^{-1})+\frac{d_{\bar{\Lambda}}d_{\Lambda}}{d(d+1)}=\mathcal{O}(d^{-1})\nonumber\,,
\end{aligned}
\end{equation}
where we used the fact that $T_{\Lambda_{12}}\tr_{2}[T_{\Lambda_1}]\tr_{1}[T_{\Lambda_2}]$ is unitary and $\norm{T_{\Lambda_{12}}\tr_{2}[T_{\Lambda_1}]\tr_{1}[T_{\Lambda_2}}_{1}]=d_{\Lambda}$, and that $d_{\Lambda}d_{\bar{\Lambda}}=d^2$, hence $|d/(d+1)-1|=\mathcal{O}(d^{-1})$. We thus find that for $\epsilon=2^{-1/3n}$, one has:
\begin{equation}
\begin{aligned}
\operatorname{Pr}\left(\mathcal{I}(\Psi^{(2)})\ge d_{\Lambda}\epsilon\right)&\le\operatorname{Pr}\left(\max_i\left|\tr\left[\Psi^{(2)}_{\Lambda}O_{\Lambda}^{(i)}\right]-\frac{\tr\left[O_{\Lambda}^{(i)}\right]}{d_{\Lambda}}\right|\ge \epsilon\right)\\&\le
\operatorname{Pr}\left(\max_i\left|\tr\left[\Psi^{(2)}O_{\Lambda}^{(i)}\right]-\overline{\tr\left[\Psi^{(2)}_{\Lambda}O_{\Lambda}^{(i)}\right]}\right|\ge \epsilon\right)\nonumber\\
&\le 2d_{\Lambda}^2e^{-C2^n\epsilon^2/4}\equiv e^{-C2^{1/3n-4}+(2n+1)\ln 2}\,,
\nonumber
\end{aligned}
\end{equation}
where we used the fact that (asymptotically):
\be
\operatorname{Pr}\left(\max_i\left|\overline{\tr\left[\Psi^{(2)}_{\Lambda}O_{\Lambda}^{(i)}\right]}-\tr[\mathcal{O}^{(i)}_{\Lambda}/d_{\Lambda}]\right|\ge 2^{-n/3}\right)=0\,.
\ee
Thus:
\be
\operatorname{Pr}\left(\mathcal{I}_{\Lambda}(\psi^{\otimes 2}\right)\le 2^{|\Lambda|-1/3n})\ge1-e^{-C2^{1/3n-4}+(2n+1)\ln 2}\,.
\ee
The calculation becomes more intricate for $k>2$. However, making the assumption that $\Lambda$ is symmetric between the copies of $\ket{\psi}$, with $|\Lambda|=k\lambda$ qubits, and setting $d_{\Lambda}\equiv 2^{k\lambda}$, we can use 
\be
\frac{\tr_{\bar{\Lambda}}[\Pi_{\sym}^{(k)}]}{D_{\sym}^{(k)}}=\frac{\mathds{1}_{\Lambda}2^{(n-\lambda) k}}{D_{\sym}}+\sum_{c(\pi)}T_{\pi}\frac{2^{(n-\lambda)(k-c)}}{D_{\sym}^{(k)}}\,,
\ee
where the sum is over the conjugacy class $c(\pi)$ of the symmetric group $S_{k}$. As evidenced by the above formula, the order of the trace depends on the conjugacy class $c$. For example for $c(\pi)=(12), (23),\ldots$ one has $c=1$, for $c(\pi)=(1234),(1243),\ldots$ one has $c=3$ and so on. Finally:
\begin{equation}
\begin{aligned}
\norm{\frac{\tr_{\bar{\Lambda}}[\Pi_{\sym}^{(k)}]}{D_{\sym}^{(k)}}-\frac{\mathds{1}_{\Lambda}}{d_{\Lambda}}}&\le \left|\frac{2^{(n-\lambda)}}{D_{\sym}^{(k)}}\right|+(k!-1)\frac{2^{(n-\lambda)(k-1)}}{D_{\sym}}\\&=\frac{1}{2^{k\lambda}}-\frac{1}{d_{\Lambda}}+\mathcal{O}(2^{-n})+\mathcal{O}((k-1)!2^{-n})\\
&=\mathcal{O}(2^{-n})\,.
\end{aligned}
\end{equation}
Thus, we have:
\begin{equation}
\operatorname{Pr}\left(\left|\tr\left[\Psi^{(k)}_{\Lambda}\right]-\frac{\tr\left[O_{\Lambda}^{(i)}\right]}{d_{\Lambda}}\right|\ge \epsilon\right)\le
\operatorname{Pr}\left(\left|\tr\left[\Psi^{(k)}_{\Lambda}\right]-\overline{\tr\left[\Psi^{(k)}_{\Lambda}O_{\Lambda}^{(i)}\right]}\right|\ge \epsilon\right)\,,
\end{equation}
provided that $\epsilon>\mathcal{O}(2^{-n})$. By applying Levy's lemma thanks to the typicality of $\tr[\Psi^{(k)}_{\Lambda}]$, choosing $\epsilon=2^{-1/3n}$ we have:
\be
\operatorname{Pr}(\mathcal{I}_{\Lambda}(\Psi^{(k)})\le 2^{-1/3n})\ge1-e^{\tilde{C}2^{n/3}/2k+\lambda}\,.
\ee

\section{Proof of Theorem~\ref{theore:anticoncentration} and Proposition~\ref{eq:prop-anticoncentration}}\label{App: Theorem 2}
\subsection{Proof of Theorem~\ref{theore:anticoncentration}}
Let us consider the following quantity
\begin{equation}
    \Delta f_{A,B}=f(\vec{\theta}_A) - f(\vec{\theta}_B)\,,
\end{equation}
where $\vec{\theta}_A$ and $\vec{\theta}_B$ are such that $\vec{\theta_A}=\vec{\theta_B}+\hat{\vec{e}}_{AB}l_{AB}$ with $l_{AB}\in\Omega(1/\poly(n))$.  Our goal will be to show that the variance of this quantity is at most polynomially vanishing. 

First, we will use the following notation $\vec{\theta}_n\equiv \vec{\theta}_A$ and $\vec{\theta}_0\equiv \vec{\theta}_B$.  Moreover, it is useful to further divide the path in  the parameter space by a sequence of single parameter changes so that
\begin{equation}
    \vec{\theta}_0\rightarrow \vec{\theta}_1\rightarrow\ldots\rightarrow \vec{\theta}_m\,,
\end{equation}
where $m\in\mathcal{O}(\poly(n))$. Here we have defined 
\begin{equation}
    \vec{\theta}_{i+1}=\vec{\theta}_0+\sum_{j=1}^{i} \hat{\vec{e}}_{j} l_{j}=\vec{\theta}_{i}+ \hat{\vec{e}}_{i} l_{i}\,,
\end{equation}
where $\hat{\vec{e}}_{i}$ is a vector with a single one and where at least one $l_i$ is such that $\sin^2(l_{i})\in\Omega(1/\poly(n))$. Note that we can guarantee both $m\in\mathcal{O}(\poly(n))$ and that there exists an $l_i$ such that $\sin^2(l_{i})\in\Omega(1/\poly(n))$ from the fact that we have at most a polynomial number of parameters, and that $\vec{\theta}_A$ and $\vec{\theta_B}$ are at most polynomially close. The previous allows us to note that
\begin{equation}
    \Delta f_{n,0}=\sum_{i=0}^{m-1} f(\vec{\theta}_{i+1}) - f(\vec{\theta}_{i})=\sum_{i=0}^{m-1} \Delta  f_{i+1,i}\,.\label{eq:deltafs}
\end{equation}
where we defined $\Delta f_{i+1,i}\coloneqq f(\vec{\theta}_{i+1})-f(\vec{\theta}_i)$. Taking the variance of Eq.~\eqref{eq:deltafs} we have
\begin{equation}
\begin{aligned}
    \Var_{\vec{\theta}_0}[\Delta f_{n,0}]&=\Var_{\vec{\theta}_0}[\sum_{i=0}^{m-1} \Delta  f_{i+1,i}] \\
    &= \sum_{i,j=0}^{n-1} \Cov[ \Delta  f_{i+1,i},\Delta  f_{j+1,j}]\,.\label{eq:lower}
\end{aligned}
\end{equation}
Here we recall that $\Cov[ \Delta  f_{i+1,i},\Delta  f_{i+1,i}]=\Var_{\vec{\theta}_0}[\Delta  f_{i+1,i}]$. 

In what follows we will use the following lemmas, which we will prove at the end of this appendix( See App.~\ref{App:lem}).

\begin{lemma}\label{lem:p11}
Let $U(\vec{\theta})$ be a shallow HEA with depth $D\in\mathcal{O}(\log(n))$, and $O=\sum_{i}c_i\bigotimes_{\alpha}O_{\alpha}^{(i)}$ with $O_{\alpha}^{(i)}$ being traceless operators having support on at most two neighboring qubits. Then, 
\begin{equation}
    \mathbb{E}_{\vec{\theta}_0}[f(\vec{\theta}_{i})]=0\,,\quad \forall i\,.
\end{equation}
\end{lemma}

\begin{lemma}\label{lem:p12}
Let $U(\vec{\theta})$ be a shallow HEA with depth $D\in\mathcal{O}(\log(n))$, and $O=\sum_{i}c_i\bigotimes_{\alpha}O_{\alpha}^{(i)}$ with $O_{\alpha}^{(i)}$ being traceless operators having support on at most two neighboring qubits. Then, 
\begin{equation}
    \mathbb{E}_{\vec{\theta}_0}[\Delta f_{i+1,i}\Delta f_{j+1,j}]=0\,,\quad \forall i\neq j\,.
\end{equation}
\end{lemma}

\begin{lemma}\label{lem:p13}
Let $U(\vec{\theta})$ be a shallow HEA with depth $D\in\mathcal{O}(\log(n))$ where each local two-qubit gates forms a $2$-design on two qubits, and let $O=\sum_{i}c_i\bigotimes_{\alpha}O_{\alpha}^{(i)}$ be the measurement composed of, at most, polynomially many traceless Pauli operators  $O_{\alpha}^{(i)}$ having support on at most two neighboring qubits, and where $\sum_i c_i^2\in\mathcal{O}(\poly(n))$. Then, if the input state follows an area law of entanglement, we have 
\be
\Var[\partial_{\mu}f(\vec{\theta})]\in\Omega\left(\frac{1}{\poly(n)}\right)\,.
\ee
\end{lemma}

Let us now go back to Eq.~\eqref{eq:lower}. Note that 
\begin{equation}
\begin{aligned}
    \Cov[ \Delta  f_{i+1,i},\Delta  f_{j+1,j}]=&\mathbb{E}[\Delta  f_{i+1,i}\Delta  f_{j+1,j}]-\mathbb{E}[\Delta  f_{i+1,i}]\mathbb{E}[\Delta  f_{j+1,j}]=0
\end{aligned}
\end{equation}
where in the second line we have used Lemmas~\ref{lem:p11} and~\ref{lem:p12}. Thus, we have 
\begin{align}
    \Var_{\vec{\theta}_0}[\Delta f_{n,0}]
    &= \sum_{i=0}^{m-1} \Var[ \Delta  f_{i+1,i}]\,.
\end{align}\label{eq:lower23}
Using the  fact that $\vec{\theta}_{i+1}=\vec{\theta}_{i}+ \hat{\vec{e}}_{i} l_{i}$, and leveraging the  parameter shift-rule for computing gradients~\cite{mitarai2018quantum,schuld2019evaluating}, we then get 
\begin{equation}
    \Delta  f_{i+1,i}=f(\vec{\theta}_{i+1})-f(\vec{\theta}_{i})=2\sin\left(\frac{l_i}{2}\right)\partial_i f(\vec{\theta}_i)\,.
\end{equation}
and 
\begin{align}
    \Var_{\vec{\theta}_0}[\Delta f_{n,0}]
    &= \sum_{i=0}^{m-1}4\sin^2\left(\frac{l_i}{2}\right) \Var[ \partial_i f(\hat{\vec{\theta}}_i)]\,,
\end{align}
where we have defined  $\hat{\vec{\theta}}_{i}=\vec{\theta}_{i-1}+ \hat{\vec{e}}_{i} \frac{l_{i}}{2}$. Then, recalling that $m\in\mathcal{O}(\poly(n))$ and that there exists an $l_i$ such that $\sin^2(l_{i})\in\Omega(1/\poly(n))$, we can use Lemma~\ref{lem:p13} to find 
\begin{equation}
    \Var_{\vec{\theta}_0}[\Delta f_{n,0}]\in\Omega\left(\frac{1}{\poly(n)}\right)\,.
\end{equation}
\subsection{Proof of Proposition~\ref{eq:prop-anticoncentration}}
Here we note that in the previous section, where we have proved Theorem~\ref{theore:anticoncentration} we have shown that the absence of barren plateaus (through Lemma~\ref{lem:p13}) implies cost values anti-concentration. In this section, we prove the converse.

First, we note that  by anti-concentration we mean that for any set of parameters $\vec{\theta_B}$ and $\vec{\theta_A}=\vec{\theta_B}+\hat{\vec{e}}_{AB}l_{AB}$ with $l_{AB}\in\Omega(1/\poly(n))$ one has
\be\label{eq:anticon}
 \Var_{\vec{\theta}_B}[f(\vec{\theta}_A)-f(\vec{\theta}_B)]\in\Omega\left(\frac{1}{\poly(n)}\right)\,.
\ee 
Then, let us use the fact that  
\begin{equation}
    \Var_{\vec{\theta}}[\partial_\nu f(\vec{\theta})]=\frac{1}{4}\Var[f(\vec{\theta}^+)-f(\vec{\theta}^-)]\,,
\end{equation}
where we have used the parameter-shift rule, and where $\vec{\theta}^+=\vec{\theta}^+\pm \hat{\vec{e}}_\nu\frac{\pi}{2}$ such that $\hat{\vec{e}}_\nu$ is a unit vector with a one at the $\nu$-th entry. Then, we have that since the difference between $\vec{\theta}^+$ and $\vec{\theta}^-$ is $\mathcal{O}(1)$, we can use Eq.~\eqref{eq:anticon} to find 
\begin{equation}
    \Var_{\vec{\theta}}[\partial_\nu f(\theta)]\in\Omega\left(\frac{1}{\poly(n)}\right)\,,
\end{equation}
which completes the proof of  Proposition~\ref{eq:prop-anticoncentration}.

\section{Isospectral twirling and proof of Proposition~\ref{propgde1},Thereom~\ref{gdehamtheorem} and Proposition~\ref{propgde2}\label{App: isospec}}
In this section, we aim to prove Proposition~\ref{propgde1}, Theorem~\ref{gdehamtheorem}, Corollary~\ref{gdehamcor}, and Proposition~\ref{propgde2}.
\subsection{Isospectral twirling}
We first, review the useful notion of the Isospectral twirling introduced in Refs.~\cite{leone2020isospectral,Oliviero2020random}. Consider an Hamiltonian $H$ written in its spectral decomposition $H=\sum_{k}E_{k}\Pi_{k}$, where $\Pi_{k}$ are its eigenvectors and $E_{k}$ are its eigenvalues. Consider the time-evolution generated by $H$, i.e., $W(t)= \exp(-iHt)=\sum_{k}e^{-iE_{k}t}\Pi_{k}$. Denote $\mathbb{U}(n)$ the unitary group on $n$ qubits. Define the following ensemble of isospectral unitary evolutions:
\be
\mathcal{E}_{H}=\{G^{\dag}\expf{-iHt}G\,|\, G\in\mathbb{U}(n)\}
\ee
whose representative element is $H$. The Isospectral twirling of order $k$, denoted as $\mathcal{R}^{(2k)}(W(t))$ is the $2k$-fold Haar channel of the operator $W^{\otimes k,k}(t)\coloneqq W^{\otimes k}(t)\otimes W^{\dag\otimes k}(t)$:
\be
\mathcal{R}^{(2k)}(W(t))\coloneqq \int_{\mathbb{U}(n)}\de G\, G^{\dag\otimes 2k}W^{\otimes k,k}(t)G^{\otimes 2k}\,.
\ee
Using the Weingarten functions~\cite{weingarten78asym,collins2003moments,collins2006integration}, one can compute the isospectral twirling as:
\be
\mathcal{R}^{(2k)}(W(t))=\sum_{\pi\sigma\in S_{2k}}(\Omega^{-1})_{\pi\sigma}\tr[W^{\otimes k,k}(t)T_{\pi}]T_{\sigma}\,,
\ee
where $S_{2k}$ is the symmetric group of order $2k$, $T_{\pi}$ is the unitary representation of the permutation $\pi\in S_{2k}$, and $\Omega_{\pi\sigma}\coloneqq \tr[T_{\pi}T_{\sigma}]$. Note that $\tr[W^{\otimes k,k}(t)T_{\pi}]$ are spectral functions of the representative Hamiltonian $H$. $\pi=e$ (the identity) defines the $2k$-spectral form factor $c_{2k}(t)\equiv|\tr[W(t)]|^{2k}$:
\be
c_{2k}(t)=\left(\sum_{kl}e^{i(E_{k}-E_{l})t}\right)^{k}
\label{2kpointspectralformfactor}
\ee
which governs the behavior of many figures of merit of random isospectral Hamiltonians, as noted in~\cite{Oliviero2020random,leone2020isospectral}. While the expression for Isospectral twirling of order $k$, for $k\ge 2$ is cumbersome to report, we do recall  the expression for the Isospectral twirling for $k=1$:
\be
\mathcal{R}^{(2)}(t)=\frac{c_{2}(t)-1}{d^2-1}\mathds{1}^{\otimes 2}+\frac{d^2-c_{2}(t)}{d^{2}-1}\frac{T_{12}}{d}\,,
\label{2isospectraltwirling}
\ee
where $T_{12}$ is the swap operator between the two copies of $\mathcal{H}$, and $c_{2}(t)$ is the $2$-point spectral form factor in Eq.~\eqref{2kpointspectralformfactor}.
One can consider the isospectral twirling of a scalar function of the time evolution operator $W(t)$ (characterized by the operator of interest $O$), i.e., $F_{O}[W(t)]$, which can be written after the isospectral twirling as $\aver{F_{O}[G^{\dag}W(t)G]}_{G}\coloneqq \tr[T_{\sigma}O\mathcal{R}^{(2k)}(t)]$, where $T_{\sigma}$ is a particular permutation operator. Its value (depending on the evolution time $t$) characterizes the average behavior in the ensemble $\mathcal{E}_{H}$ of all those Hamiltonians sharing the same spectrum of $H$. Since the Isospectral twirling of a scalar function depends upon the particular choice of the spectrum of $H$, one then averages over spectra of a given ensemble of Hamiltonians $E$. Relevant examples are the Gaussian unitary ensemble $E\equiv \gue$, the Poisson ensemble $E\equiv \poi$, or the Gaussian diagonal ensemble $E\equiv \gde$. As one can see, in this picture, spectra and eigenvectors become completely unrelated, since the average over the full unitary group erases the information about the eigenvectors.
Although, in Refs. \cite{Oliviero2020random,leone2020isospectral} many ensembles of Hamiltonians have been considered, in this paper, we are particularly interested in the  $\gde$ ensemble which is the simplest ensemble of Hamiltonians: the $2k$-point spectral form factors can readily be computed. Let $\operatorname{sp}(H)\coloneqq \{E_{k}\}_{k=1}^{d}$, where $d\equiv 2^n$. The  $\gde$ ensemble is characterized by the following probability distribution for $\operatorname{sp}(H)$:
\be
P_{\gde}(\{E_{k}\})\coloneqq \left(\frac{2}{\pi}\right)^{d/2}e^{-2\sum_{k}E_{k}^{2}}\,,
\ee
i.e., all the eigenvalues $E_{k}$ are independent identically distributed Gaussian random variables with zero mean and standard deviation $1/2$. Define $\tilde{c}_{2k}(t)\coloneqq \frac{c_{2k}(t)}{d^{2k}}$, then the average of the normalized $2k$-spectral form factors $\tilde{c}_{2k}(t)$ reads
\be
\overline{\tilde{c}_{2k}(t)}^{\gde}=e^{-kt^2/4}+\mathcal{O}(d^{-1})\,.
\label{spectralform2kgde}
\ee
Let us now build a notation useful throughout the following proofs. Let $F[W(t)]$ be a scalar function of the unitary evolution $W(t)=\expf{-iHt}$. Then we denote with $\mathbb{E}_{\gde}$ the Isospectral twirling of the scalar function $F[W(t)]$ followed by the average over the $\gde$ ensemble of Hamiltonian, namely:
\be
\mathbb{E}_{\gde}[F[W(t)]]\coloneqq \overline{\int\de G F[G^{\dag}W(t)G]}^{\gde}\,.
\ee
In the following section, we use the techniques introduced in order to prove Proposition~\ref{propgde1}.
\subsection{Proof of Proposition~\ref{propgde1}}
While the proof of Eq.~\eqref{symmtheta0} is straightforward, here we prove Eq.~\eqref{averagegde} via the Isospectral twirling technique. Recall $L_{s}(H_{G},\vec{\theta},t)=\tr[W(t)\st{\psi_s}W^{\dag}(t)O(\vec{\theta})]$, where $\ket{\psi_s}$ is a completely factorized state. We are interested in computing $\mathbb{E}_{\gde}[L_{s}(H_G,\vec{\theta},t)]$. Note that $L_{s}(H_{G},\vec{\theta},t)$ can be written as:
\be
L_{s}(H_{G},\vec{\theta},t)=\tr[T_{12}\psi_{s}\otimes O(\vec{\theta})W^{\otimes 1,1}(t)]\,.
\ee
The average over $\gde$ Hamiltonian can be readily taken out of Eq. \eqref{2isospectraltwirling}, and Eq. \eqref{spectralform2kgde}. By using that $\tr[O(\vec{\theta})]=0$, one gets:
\be
\mathbb{E}_{\gde}[L_{s}(H_{G},\vec{\theta}_0,t)]=e^{-t^2/4}\tr[\psi_{s}O(\vec{\theta}_0)]+\mathcal{O}(d^{-1})\,,
\ee
to recover Eq. \eqref{averagegde} it is sufficient to note that $\tr[\psi_{s}O(\vec{\theta}_0)]=1$ because $O(\vec{\theta}_0)\equiv S$, and $S\ket{\psi_s}=\ket{\psi_s}$ by definition.

\subsection{Proof of Theorem \ref{gdehamtheorem}}
To prove the theorem, we use the well-known Markov inequality: let $x$ be a positive semi-definite random variable with average $\mu$, then:
\be
\operatorname{Pr}(x\ge \epsilon)\le \frac{\mu}{\epsilon}\,.
\label{markovinequality}
\ee
To prove the concentration consider $|L_{s}(H,\vec{\theta},t)|$ for $H\in\{H_{G},H_{S}\}$ defined in the main text. Then by Jensen's inequality, we have:
\be
\mathbb{E}_{\gde}[|L_{s}(H,\vec{\theta},t)|]\le \sqrt{\mathbb{E}_{\gde}[L_{s}(H,\vec{\theta},t)^2]}\,.
\ee
Note that, one can write 
\be
L_{s}(H,\vec{\theta},t)^2=\tr[T_{(13)(24)} \psi_{s}^{\otimes 2}\otimes O(\vec{\theta})W^{\otimes 2,2}(t)]\,.
\ee
Then, taking the isospectral twirling, one has:
\be
\aver{L_{s}(G^{\dag}HG,\vec{\theta},t)^{2}}_{G}=\tr[T_{(13)(24)} \psi_{s}^{\otimes 2}\otimes O(\vec{\theta})\mathcal{R}^{(4)}(t)]\,,
\ee
where the Isospectral twirling operator $\mathcal{R}^{(4)}(t)$ can be found in Eq. $(165)$ of Ref. \cite{Oliviero2020random}. Taking the average over the $\gde$ spectra, one has
\begin{align}
\mathbb{E}_{\gde}[L_{s}(H_{G},\vec{\theta},t)^{2}]&=e^{-t^2/2}\tr[\psi_{s}O(\vec{\theta})]^2+\mathcal{O}(d^{-1})\nonumber\\
\mathbb{E}_{\gde}[L_{s}(H_{S},\vec{\theta},t)^{2}]&=e^{-t^2/2}\tr[\psi_{s}O(\vec{\theta})]^2+\mathcal{O}(d^{-1})\,,\nonumber
\end{align}
where the result is obtained after some algebra, with the hypothesis of $\psi\equiv\psi_A\otimes \psi_B$ and that $\tr [O]=0$. Note that the above result is obtained for $H\in\{H_{G},H_{S}\}$; we indeed exploited the fact that $|A|\ll |B|$, and thus $d_{B}=\mathcal{O}(d)$. Using Markov's inequality Eq. \eqref{markovinequality}, the result can be readily derived.

\subsection{Proof of Theorem \ref{propgde2}}
\subsubsection{An anti-concentration inequality}
To prove the theorem, we make use of Cantelli's inequality: let $x$ be a random variable with average $\mu$ and standard deviation $\sigma$. Then:
\be
\operatorname{Pr}(x\ge \mu-\theta\sigma)\ge \frac{\theta^2}{1+\theta^2}
\label{cantelliinequality}
\ee
To bound the measure $\mathcal{I}_{\Lambda}(\psi)$, we make use of the bound proven in~\cite{coles2019strong} between the Hilbert-Schmidt distance and the trace distance, namely:
\be
\mathcal{I}_{\Lambda}(\psi)\ge \frac{1}{2}\sqrt{\norm{\psi_{\Lambda}-\frac{\mathds{1}_{\Lambda}}{d_{\Lambda}}}_{2}}\,,
\label{hilbertschmidttracenormbound}
\ee
where $\psi_{\Lambda}\coloneqq \tr_{\bar{\Lambda}}[\st{\psi}]$, and $d_{\Lambda}=2^{|\Lambda|}$. Note that
\be
\norm{\psi_{\Lambda}-\frac{\mathds{1}_{\Lambda}}{d_{\Lambda}}}_{2}=\sqrt{P(\psi_{\Lambda})-d_{\Lambda}^{-1}}\,,
\ee
where $P(\psi_\Lambda)\equiv \tr[\psi_{\Lambda}^{2}]$ is the purity of the reduced density matrix $\psi_{\Lambda}$. Let $\ket{\psi_{t}}=\expf{-iHt}\ket{\psi_s}$ the state resulting from the time evolution under a $\gde$ Hamiltonian acting on a completely factorized state. Let $\Lambda$ be a subsystem such that $|\Lambda|=\mathcal{O}(\log (n))$. Note that if $\operatorname{Pr}(P(\psi_{\Lambda})\ge \epsilon)\ge \delta$, then
\be
\operatorname{Pr}\left(\norm{\psi_{\Lambda}-\frac{\mathds{1}_{\Lambda}}{d_{\Lambda}}}_{2}\ge \sqrt{\epsilon-d_{\Lambda}^{-1}}\right)\ge \delta\,,
\ee
and, because of Eq. \eqref{hilbertschmidttracenormbound}, one has:
\be
\operatorname{Pr}\left(\mathcal{I}_{\Lambda}(\psi)\ge \frac{1}{2}(\epsilon-d_{\Lambda}^{-1})^{1/4}\right)\ge \delta\,.
\label{mathcalIanticoncentration}
\ee
Thus, from the above chain of inequality, it is clear that it is sufficient to obtain a bound on the anti-concentration of the purity $P(\psi_{\Lambda})$.
Denoting $P_{\gde}\coloneqq \mathbb{E}_{\gde}[P(\psi_\Lambda)]$, and $Q_{\gde}=\mathbb{E}_{\gde}[P(\psi_\Lambda)^{2}]$, we can use Eq. \eqref{cantelliinequality} to write:
\be
\operatorname{Pr}\left(P(\psi_{\Lambda})\ge P_{\gde}-\theta \sqrt{Q_{\gde}-P_{\gde}^{2}}\right)\ge \frac{\theta^{2}}{1+\theta^2}\,.
\label{purityanticoncentration}
\ee
\subsubsection{Computing $P_{\gde}$}
First, let us note the following fact. $P(\psi_{\Lambda})=\tr[T_{\Lambda}\psi^{\otimes 2}_{t}]$, where $\psi_{t}=W_{G}(t)\psi_{0}W_{G}^{\dag}(t)$, and where $W_{G}(t)\equiv G^{\dag}W(t)G$. The Isospectral twirling of $P(\psi_{\Lambda})$ is the average with respect to $G\in\mathbb{U}(n)$.
\be
\aver{P(\psi_{\Lambda})}_G=\aver{\tr[T_{\Lambda}W_{G}(t)\psi_{0}^{\otimes 2}W_{G}^{\dag}(t)]}_G\,.
\ee
Thanks to the right/left invariance of the Haar measure, we can insert unitaries of the form $V\equiv V_{\Lambda}\otimes V_{\bar{\Lambda}}$ because $[V,T_{\Lambda}]=0$. This means that 
\be
\aver{P(\psi_{\Lambda})}_G=\braket{P(\psi^{(V)}_{\Lambda})}_G\,,
\label{equivalency}
\ee
where $\psi_{\Lambda}^{(V)}$ is defined as $\psi_{\Lambda}^{(V)}\coloneqq \tr_{\bar{\Lambda}}W_{G}(t)\psi^{(V)}W_{G}^{\dag}(t)$, where
\be
\psi^{(V)}\coloneqq \int_{\mathbb{U}(n_{\Lambda})}\!\!\!\!\!\!\!\!\!\!\!\de V_{\Lambda}\int_{\mathbb{U}(n_{\bar{\Lambda}})}\!\!\!\!\!\!\!\!\!\!\!\de V_{\bar{\Lambda}}(V_{\Lambda}\otimes V_{\bar{\Lambda}})^{\otimes 2}\psi_{0}(V_{\Lambda}\otimes V_{\bar{\Lambda}})^{\dag\otimes 2}\,,
\ee
i.e., we can compute the average purity over $\gde$ Hamiltonian on the average product state between $\Lambda$ and $\bar{\Lambda}$. A straightforward computation leads us to
\begin{equation}
\braket{P(\psi_{\Lambda})}_G=\frac{d(d_{\Lambda}+d_{\bar{\Lambda}})+\braket{\tr[T_{\Lambda}W_{G}^{\otimes 2}(t)T_{\Lambda}W_{G}^{\dag\otimes 2}(t)]}_{G}}{d(d_{\Lambda}+1)(d_{\bar{\Lambda}}+1)}+\frac{\braket{\tr[T_{\bar{\Lambda}}W_{G}^{\otimes 2}(t)T_{\Lambda}W_{G}^{\dag\otimes 2}(t)]}_G}{d(d_{\Lambda}+1)(d_{\bar{\Lambda}}+1)}\,.
\label{averagepurity}
\end{equation}
Following the calculations in Sec.~$3.3.1$ of Ref.~\cite{Oliviero2020random}, the result can be proven to be:
\be
P=\frac{1}{d_{\Lambda}}+e^{-t^2/2}\left(1-\frac{1}{d_{\Lambda}}\right)+\mathcal{O}(d^{-1})\,.
\label{averagepuritygde}
\ee

\subsubsection{Computing $Q_{\gde}$}
In this section, we compute the second moment of the purity, i.e. $\braket{P(\psi_{\Lambda})^2}$. Thanks to the left/right invariance of the Haar measure over $G$, and from the commutation with $T_{\Lambda}$, we can compute the average of the second moment for any factorized state $\ket{\psi_0}\equiv\ket{\psi_{\Lambda}}\otimes\ket{\psi_{\bar{\Lambda}}}$, by computing it with the input state $\ket{\psi_0}\equiv\ket{0}^{\otimes n}$. Namely:
\be
\aver{P(\psi_{\Lambda})^2}=\braket{\tr^2(T_{\Lambda}W^{\otimes 2}_{G}(t)\st{0}^{\otimes 2n}W^{\dag\otimes 2}_{G})(t)}_G
\label{messyP2formula}
\ee
Computing the above expression is trickier than computing $P_{\gde}$ because the latter involves averaging over the $8$-th fold power of $G$. It is well known that the permutation group $S_{8}$ contains $8!$ elements, making the calculation too expensive. It is also known that purity and OTOCs possess many similarities (see Refs.~\cite{hosur2016chaos,leone2021quantum,oliviero2021transitions,ding2016conditional,liu2018entanglement}). Indeed, the strategy to do the calculation will be: $(i)$ reduce each term of Eq. \eqref{messyP2formula} to a sum of high-order OTOCs, and $(ii)$ use the following asymptotic formula proven in~\cite{cotler2017chaos}: let $A_{1},\ldots, A_{k},B_{1},\ldots, B_{k}$ non-identity Pauli operators, and let $B_{l}(t)\equiv W_{G}^{\dag}(t)B_{l}W_G(t)$:
\be
\aver{\tr\left[\prod_{l}A_{l}B_{l}(t)\right]}_{G}=\tr\left[\prod_{l}A_{l}B_{l}\right]\tilde{c}_{2k}(t)+\mathcal{O}(d^{-1})\,,
\label{cotler}
\ee
where $\tilde{c}_{2k}(t)$ is the normalized $2k$-spectral form factor defined in Eq.~\eqref{2kpointspectralformfactor}. The strategy is thus to write \eqref{messyP2formula} in terms of $8-$OTOCs. First, note that we can write the swap operator as:
\be
T_{\Lambda}=\frac{\mathds{1}_{\Lambda}}{d_{\Lambda}}+\frac{1}{d_{\Lambda}}\sum_{P_{\Lambda}\neq \mathds{1}}P_{\Lambda}^{\otimes 2}\,,
\ee
and, in the same fashion, we can write the state $\st{0}$ as:
\be
\st{0}=\sum_{P\in\{\mathds{1},Z\}^n}P\,.
\ee
Thus, we can write:
\begin{align}
P^{2}(\psi_{\Lambda})=\!\!\!\!\!\!\sum_{P_{i},P_{\Lambda}^{a},P_{\Lambda}^{b}}\!\!\!\!\!\!\tr[P_{\Lambda}^{a}\tilde{P}_{1}]\tr[P_{\Lambda}^{a}\tilde{P}_{2}]\tr[P_{\Lambda}^{b}\tilde{P}_{3}]\tr[P_{\Lambda}^{b}\tilde{P}_{4}]\,.\nonumber
\end{align}
Note that the sum is over $P_{i}$ for $i=1,\ldots, 4$ and runs over the subgroup $P\in\{\mathds{1},Z\}^n$; moreover we defined $\tilde{P}_{i}\coloneqq U_{G}P_{i}U_{G}^{\dag}$ for $i=1,\ldots, 4$ to light the notation. It is useful for the following to split the identity part of the sum, and write $P^{2}(\psi_{\Lambda})$ as:
\begin{equation}
P^{2}(\psi_{\Lambda})=\sum_{P_{i},P_{\Lambda}^{a},P_{\Lambda}^{b}\neq\mathds{1}}\tr[P_{\Lambda}^{a}\tilde{P}_{1}]\tr[P_{\Lambda}^{a}\tilde{P}_{2}]\tr[P_{\Lambda}^{b}\tilde{P}_{3}]\tr[P_{\Lambda}^{b}\tilde{P}_{4}]+\frac{1}{d_A}P(\psi_{\Lambda})-\frac{2}{d_{A}^2}\,,
\end{equation}
where each term on the first sum is different from $\mathds{1}$. To write it in terms of OTOCs, let us use the following identity: let $A$ and $B$ two operators, and $P$ Pauli operators then:
\be
\frac{1}{d}\sum_{P}\tr[APBP]=\tr[A]\tr[B]\,,
\ee
where the summation is on the whole Pauli group. Thus
\begin{equation}
P^{2}(\psi_{\Lambda})=\sum_{\substack{P_{i},P_{\Lambda}^{a},P_{\Lambda}^{b}\neq \mathds{1}\\ Q,K,L}}\!\!\!\!\!\!\tr\left[P_{\Lambda}^{a}\tilde{P}_{1}QP_{\Lambda}^{a}\tilde{P}_{2}QKP_{\Lambda}^{b}\tilde{P}_{3}KLP_{\Lambda}^{b}\tilde{P}_{4}L\right]+\frac{1}{d_A}P(\psi_{\Lambda})-\frac{2}{d_{A}^2}
\label{e1}\,.
\end{equation}
Here, $Q,K,L$ are labeling global Pauli operators, running on the whole Pauli group. In order to recover OTOCs, we still need to split the sum of $Q,K,L$ between the identity and the other non-identity Paulis: this operation gives rise to $8$ different terms, labeled $t_{1},\ldots, t_{8}$, each of which proportional to the spectral function $\tilde{c}_{8}(t)$ thanks to Eq. \eqref{cotler}. The coefficients are listed below:
\begin{align}
t_{1}&\equiv\sum\tr[P_{\Lambda}^{a}P_1P_{\Lambda}^{a}P_2P_{\Lambda}^{b}P_3P_{\Lambda}^{b}P_4]\\\nonumber
t_{2}&\equiv\sum\tr[P_{\Lambda}^{a}P_1QP_{\Lambda}^{a}P_2QP_{\Lambda}^{b}P_3P_{\Lambda}^{b}P_4]\nonumber\\
t_{3}&\equiv\sum\tr[P_{\Lambda}^{a}P_1P_{\Lambda}^{a}P_2KP_{\Lambda}^{b}P_3KP_{\Lambda}^{b}P_4]\nonumber\\
t_{4}&\equiv\sum\tr[P_{\Lambda}^{a}P_1P_{\Lambda}^{a}P_2P_{\Lambda}^{b}P_3LP_{\Lambda}^{b}P_4L]\nonumber\\
t_{5}&\equiv\sum\tr[P_{\Lambda}^{a}P_1QP_{\Lambda}^{a}P_2QKP_{\Lambda}^{b}P_3KP_{\Lambda}^{b}P_4]\nonumber\\
t_{6}&\equiv\sum\tr[P_{\Lambda}^{a}P_1QP_{\Lambda}^{a}P_2QP_{\Lambda}^{b}P_3LP_{\Lambda}^{b}P_4L]\nonumber\\
t_{7}&\equiv\sum\tr[P_{\Lambda}^{a}P_1P_{\Lambda}^{a}P_2KP_{\Lambda}^{b}P_3KLP_{\Lambda}^{b}P_4L]\nonumber\\
t_{8}&\equiv\sum\tr[P_{\Lambda}^{a}P_1QP_{\Lambda}^{a}P_2QKP_{\Lambda}^{b}P_3KLP_{\Lambda}^{b}P_4L]\nonumber
\end{align}
where we adopted the following convention: the sum $\sum$ contains the sum over all the Pauli's appearing on the summation with the exception on the identity, running on their respective support; more precisely, $P_{1},\ldots, P_{4}$ runs in the subgroup $P_{i}\in\{\mathds{1},Z\}^{n}$ but the identity, $Q,K,L$ run over the whole Pauli group but the identity, while $P_{\Lambda}^{a}, P_{\Lambda}^{b}$ run over the whole Pauli group defined on $\Lambda$ qubits but the identity.
Using Eq. \eqref{averagepuritygde}, from Eq. \eqref{e1}, one thus arrives to
\be
\braket{P^{2}(\psi_{\Lambda})}_{G}=-\frac{1}{d_{\Lambda}^{2}}+2\tilde{c}_{4}\frac{d_{\Lambda}-1}{d_{\Lambda}^{2}}+\tilde{c}_{8}(t)\sum_{i=1}^{8}t_{i}+\mathcal{O}(d^{-1})\,.\nonumber
\ee
After further algebraic simplifications, one gets:
\be
\sum_{i=1}^{8}t_{i}=\frac{(d_{\Lambda}-1)^{2}}{d_{\Lambda}^{2}}+\mathcal{O}(d^{-2})\,.
\ee
In order to properly use Eq. \eqref{cotler}, one needs to be sure to absorb all the term $\mathcal{O}(d^{-1})$. The final result is:
\be
\braket{P^{2}(\psi_{\Lambda})}=\frac{1+2\tilde{c}_{4}(d_{\Lambda}-1)+\tilde{c}_{8}(d_{\Lambda}-1)^{2}}{d_{\Lambda}^{2}}+\mathcal{O}(d^{-1})\,.
\label{averagesquarepurity}
\ee
Hence, we can readily compute $Q_{\gde}$ by using Eq. \eqref{spectralform2kgde}, and obtain:
\be
Q_{\gde}=\frac{1+2e^{-t^2/2}(d_{\Lambda}-1)+e^{-t^2}(d_{\Lambda}-1)^{2}}{d_{\Lambda}^{2}}+\mathcal{O}(d^{-1})\,.
\label{averagesquarepuritygde}
\ee
\subsubsection{Final bound}
Putting Eq.~\eqref{averagepurity} and Eq.~\eqref{averagesquarepurity} together, we obtain the fluctuations of the purity with respect to an isospectral ensemble of Hamiltonian as a function of $t$:
\be
\Delta_G [P(\psi_{\Lambda})]=(\tilde{c}_{8}(t)-\tilde{c}_{4}^{2}(t))\frac{(d_{\Lambda}-1)^{2}}{d_{\Lambda}^2}+\mathcal{O}(d^{-1})\,,
\ee
where $\Delta_G [P(\psi_{\Lambda})]\coloneqq \braket{P^{2}(\psi_\Lambda)}_G-\braket{P(\psi_\Lambda)}_{G}^{2}$. Note that the above result is completely general, and valid for any isospectral family of Hamiltonians. Further, substituting Eq.~\eqref{averagepuritygde} and Eq. \eqref{averagesquarepuritygde} we get that the purity fluctuations with respect to $\gde$ Hamiltonians are
\be
Q_{\gde}-P_{\gde}^{2}=\mathcal{O}(d^{-1})\,.
\label{fluctuationspuritygde}
\ee
By using Eq. \eqref{purityanticoncentration}, and the fact that the fluctuations in Eq.~\eqref{fluctuationspuritygde} are $\mathcal{O}(d^{-1})$, we can choose $\theta=\mathcal{O}(d^{1/4})$, to write the following anti-concentration bound:
\begin{align}
\operatorname{Pr}[P(\psi_{\Lambda})&\ge d_{\Lambda}^{-1}+e^{-t^2/2}(1-d_{\Lambda}^{-1})/2+\mathcal{O}(d^{-1/4})]\nonumber\\&\ge1-\mathcal{O}(d^{-1/2}) \,.
\end{align}
Now, we can finally use Eq. \eqref{mathcalIanticoncentration} to prove the theorem:
\be
\operatorname{Pr}(\mathcal{I}_{\Lambda}(\psi)\ge e^{-t^2/8}(1-d_{\Lambda}^{-1})^{1/4}/2)\ge 1-\mathcal{O}(d^{-1/2})\,. 
\ee
Thus choosing $t=\mathcal{O}(\sqrt{\log (n)})$, and recalling that $|\Lambda|=\mathcal{O}(\log (n))$:
\be
\operatorname{Pr}\left[\mathcal{I}_{\Lambda}(\psi)\in\Omega\left(\frac{1}{\poly(n)}\right)\right]\ge 1-\mathcal{O}(2^{-n/2}) \,,
\ee
which concludes the proof.

\section{Useful Lemmas\label{App:lem}}
Before proceeding to prove Lemmas, we recall that 
if a given two-qubit gate $V$ in the HEA forms a $2$-design, one can employ the following element-wise formula of the Weingarten calculus in Refs.~\cite{collins2006integration,puchala2017symbolic} to explicitly evaluate averages over $V$ up to the second moment:
{\small
\begin{equation}
\begin{aligned}
    \int dV v_{\vec{i}\vec{j}}v_{\vec{i}'\vec{j}'}^*&=\frac{\delta_{\vec{i}\vec{i}'}\delta_{\vec{j}\vec{j}'}}{4}        \label{eq:Haar2moment}\\
\!\!\int dV v_{\vec{i}_1\vec{j}_1}v_{\vec{i}_2\vec{j}_2}v_{\vec{i}_1'\vec{j}_1'}^{*}v_{\vec{i}_2'\vec{j}_2'}^{*}&=\frac{1}{15}\left(\Delta_1-\frac{\Delta_2}{4}\right)
\end{aligned}
\end{equation}}
where $v_{\vec{i}\vec{j}}$ are the matrix elements of $V$, and
\begin{equation}
\begin{aligned}
\Delta_1&=\delta_{\vec{i}_1\vec{i}_1'}\delta_{\vec{i}_2\vec{i}_2'}\delta_{\vec{j}_1\vec{j}_1'}\delta_{\vec{j}_2\vec{j}_2'}+\delta_{\vec{i}_1\vec{i}_2'}\delta_{\vec{i}_2\vec{i}_1'}\delta_{\vec{j}_1\vec{j}_2'}\delta_{\vec{j}_2\vec{j}_1'}\,,\\
\Delta_2&=\delta_{\vec{i}_1\vec{i}_1'}\delta_{\vec{i}_2\vec{i}_2'}\delta_{\vec{j}_1\vec{j}_2'}\delta_{\vec{j}_2\vec{j}_1'}+\delta_{\vec{i}_1\vec{i}_2'}\delta_{\vec{i}_2\vec{i}_1'}\delta_{\vec{j}_1\vec{j}_1'}\delta_{\vec{j}_2\vec{j}_2'}\,.    \label{eq:Haar2moment2}
\end{aligned}
\end{equation}
Here the integration is taken over $\mathbb{U}(2)$, the unitary group of degree $2$. 
\setcounter{lemma}{0}
\setcounter{corollary}{4}
\begin{lemma}\label{lemma1b}
Let $U(\vec{\theta})$ be a HEA with depth $D$ as in Fig.~\ref{fig:HEA}(b), and $O=\sum_{i}c_i P_i=\sum_{i}c_i\bigotimes_{\alpha}O_{\alpha}^{(i)}$ be an operator, as in Eq.~\eqref{operatordef}. Then for any $O_{\alpha}^{(i)}$:
\be
|\supp(U(\vec\theta)O_{\alpha}^{(i)}U^{\dag}(\vec\theta))|\le \sum_{\substack{k\in C_{\alpha}^{(i)}\\ l_{k}^{(i)}< 2D}} l_{k}^{(i)}+2D\,.
\ee
\begin{proof}
Given $O=\sum_{i}c_{i}P_{i}=\sum_{i}c_{i}\bigotimes_{\alpha}O_{\alpha}^{(i)}$ (from the clustering of qubits), to prove the statement, one has to look at a given operator $O_{\alpha}^{(i)}$, having support on the cluster $C_{\alpha}^{(i)}$ defined in Sec.~\ref{App: proofth1}. 

As a starting point, let us introduce the local Pauli operator $\sigma_{j}^{(i)}$ possessing support only on the $j-$qubit, and then write $U(\vec{\theta})$ as in Eq.~\eqref{HEAequation} layer by layer, as 
\be
U(\vec{\theta})=V_{1}\ldots V_{D-1}V_{D}
\ee
where $V_{k}$ are the unitaries acting on each layer $k=1,\ldots, D$. For sake of simplicity, we drop the parameter dependence. Each $V_{k}$
can be further decomposed as:
\be
V_{k}=V_{1}^{(k)}\otimes\ldots \otimes V_{\frac{n}{2m}}^{(k)}\otimes V_{\frac{n}{2m}-1}^{(k)}
\ee
where each unitaries $V_{\alpha}$, $\alpha=1,\ldots, n/2m$, acts on $m$ qubits, being $n$ a multiple of $m$ ($m$ even), with periodic boundary conditions; so that either $(i)$ $V_{1}$ acts on the first $m$ qubits, $V_{2}$ acts on the second $m$ qubits, up to $V_{n/2m}$ acting on the last $m$ qubits; or $(ii)$ $V_{1}$ is acting from qubit $m/2+1$ to qubit $3m/2$, $V_{2}$ is acting from qubit $3m/2+1$ to qubit $2m$, up to $V_{n/2m}$ acting from qubit $n-m/2+1$ to qubit $m/2$. At the first layer, only one $V_{\alpha}^{(1)}$ for some $\alpha$ is acting on $\sigma_{j}^{(i)}$:
\be
V_{1}^{\dag}\sigma_{j}^{(i)}V_{1}=V_{\alpha}^{(1)\dag}\sigma_{j}^{(i)}V_{\alpha}^{(1)}
\ee
increasing the support to (at most) $m$ qubits, $\supp(V_{\alpha}^{(k)\dag}P_{i}V_{\alpha}^{(k)})\le m$. At the second layer, there will be (at most) two $V_{\beta}^{(2)}$, $V_{\gamma}^{(2)}$ for some $\beta,\gamma$ acting non-trivially on $V_{\alpha}^{(1)\dag}\sigma_{j}^{(i)}V_{\alpha}^{(1)}$:
\be
V_{2}^{\dag}V_{1}^{\dag}\sigma_{j}^{(i)}V_{1}V_{2}^{\dag}=V_{\gamma}^{(2)\dag}V_{\beta}^{(2)\dag}V_{\alpha}^{(1)\dag}\sigma_{j}^{(i)}V_{\alpha}^{(1)}V_{\beta}^{(2)}V_{\gamma}^{(2)}
\ee
increasing the support to at most $2m$. After the second layer, more unitaries will act on it trivially, but the support can only increase by a factor $m$ at each layer. This is because, at the $k$-th layer, only the unitaries acting on the boundary of $\supp(V_{k-1}^{\dag}\sigma_{j}^{(i)} V_{k-1})$ can increase the size of a factor $m/2$ each. Having shown the statement for an ultra-local operator, now the proof follows easily. Indeed, $O_{\alpha}$ has support on the cluster $C_{\alpha}$, which contains non-first-neighboring qubits whose relative distance is less than $mD$. The support of each operator living on the $q-$th qubit of $C_{\alpha}$ will increase by a factor $mD$ after the action of $U(\vec{\theta})$, which means that the supports are overlapping after the evolution, and thus one can upper-bound the total support of $U^{\dag}(\vec{\theta})O_{\alpha}U(\vec{\theta})$ as:
\be
|\supp(U^{\dag}(\vec{\theta})O_{\alpha}U(\vec{\theta}))|\le \sum_{k\in \mathcal{C}_{\alpha}}l_{k}+mD
\ee
i.e. by summing all the relative distances between the qubits within $C_{\alpha}$. This concludes the proof.    
\end{proof}
\end{lemma}
\begin{corollary}\label{corollary1b}
Let $U(\vec{\theta})$ be HEA with depth $D$, with $D\in \mathcal{O}(\log (n))$ and $O=\sum_{i}c_{i}P_{i}=\sum_{i}c_i\bigotimes_{\alpha}O_{\alpha}^{(i)}$. Then for any $O_{\alpha}^{(i)},O_{\alpha^\prime}^{(i)}$, with $\alpha\neq \alpha^{\prime}$, one finds that asymptotically, 
\be
\supp(U(\vec\theta)O_{\alpha}^{(i)}U^{\dag}(\vec\theta))\cap \supp(U(\vec\theta)O_{\alpha^\prime}^{(i)}U^{\dag}(\vec\theta))=\emptyset\,.
\ee
\begin{proof}
The corollary easily descends from Lemma~\ref{lemma1b} and relative proof. Indeed, each operator $O_{\alpha}$ has support on a cluster $C_{\alpha}$. This means that two different operators $O_{\alpha},O_{\alpha^\prime}$ have support with relative distance greater than $mD$. Since the support of each operator $O_{\alpha},O_{\alpha^\prime}$ can increase no more than $2D$ away from any qubit $q\in C_{\alpha}$, we have:
\be
\supp(U^{\dag}(\vec{\theta})O_{\alpha}U(\vec{\theta}))\cap\supp(U^{\dag}(\vec{\theta})O_{\alpha^{\prime}}U(\vec{\theta}))=\emptyset
\ee
by construction. It is important to remark that the above result is only valid asymptotically. Indeed, for small $n$, say $n_{0}$, it can be the case that $n\in\mathcal{O}(\poly( \log (n)))$. This concludes the proof.
\end{proof}
\end{corollary}
\begin{lemma}\label{lem:p1}
Let $U(\vec{\theta})$ be a shallow HEA with depth $D\in\mathcal{O}(\log(n))$, and $O=\sum_{i}c_i\bigotimes_{\alpha}O_{\alpha}^{(i)}$ with $O_{\alpha}^{(i)}$ being traceless operators having support on at most two neighboring qubits. Then, 
\begin{equation}
    \mathbb{E}_{\vec{\theta}_0}[f(\vec{\theta}_{i})]=0\,,\quad \forall i\,.
\end{equation}
\end{lemma}
\begin{proof}
Let us recall that $f(\vec{\theta}_{i})=\Tr[U(\vec{\theta}_{i})\st{\psi}U^{\dag}(\vec{\theta}_{i})O]$, where $U(\vec{\theta}_{i})$ is a HEA where each two-qubit $V_i$  forms a local $2$-design  (see Fig.~\ref{fig:HEA}). From the previous, we know that
\begin{equation}
    \mathbb{E}_{\vec{\theta}_0}[\cdot ]=\prod_{V_i}\int_{\mathbb{U}(2)}d\mu (V_i) (\cdot ) \,.
\end{equation}

\begin{figure*}[h!t]
    \centering
    \includegraphics[width=1\linewidth]{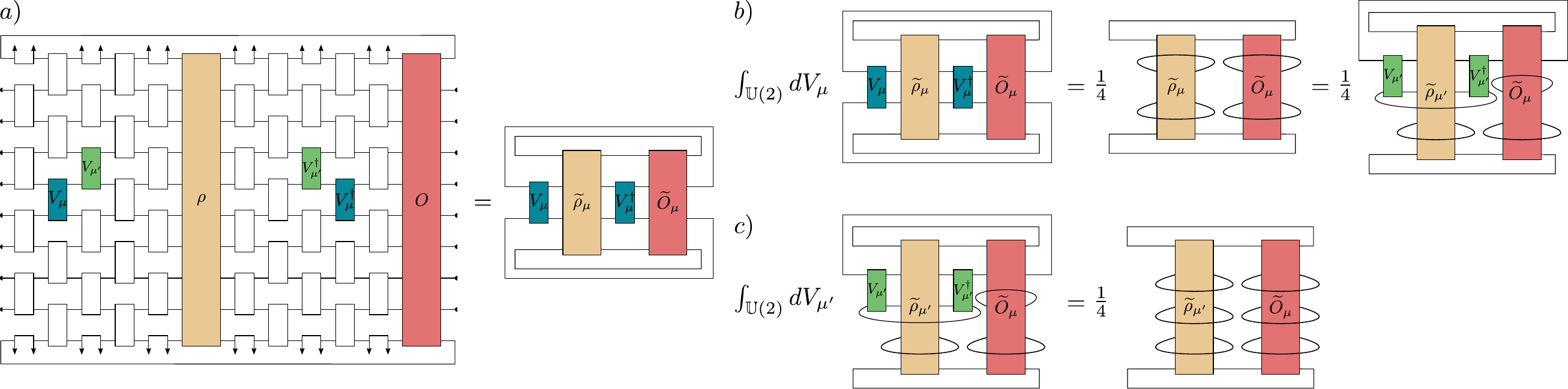}
    \caption{a) Tensor network representation of the function $f(\vec{\theta}_{i})=\Tr[U(\vec{\theta}_{i})\st{\psi}U^{\dag}(\vec{\theta}_{i})O]$. In (b) and (c) we show the results obtained by averaging over $V_\mu$ and $V_{\mu'}$, assuming that they form $2$-designs.}
    \label{fig:proof-1}
\end{figure*}

In particular, given a single two-qubit gate $V_\mu$, let us compute 
\begin{equation}
    \int_{\mathbb{U}(2)}dV_\mu  f(\hat{\vec{\theta}}_{i})=\int_{\mathbb{U}(2)}dV_\mu \Tr[V_\mu \widetilde{\rho}_\mu V_\mu ^\dagger \widetilde{O}_\mu]\,,\label{eq:tensor-con-1}
\end{equation}
where we have defined
\begin{align}
\widetilde{\rho}_\mu&=U_B^\mu \st{\psi} (U_B^\mu)^\dagger\\
\widetilde{O}_\mu&=(U_A^\mu)^\dagger O U_A^\mu
\end{align}
 with $U(\vec{\theta})=U_A^\mu V_\mu U_B^\mu$. In Fig.~\ref{fig:proof-1}(a) we schematically show the tensor representation of contraction of Eq.~\eqref{eq:tensor-con-1}. A direct calculation using Eq.~\eqref{eq:Haar2moment} shows that 
\begin{equation}
    \int_{\mathbb{U}(2)}dV_\mu  f(\vec{\theta}_{i})=\frac{1}{4}\Tr[\Tr_{\mu}[\widetilde{\rho}_\mu]\Tr_{\mu}[\widetilde{O}_\mu]]\,,
\end{equation}
where $\Tr_{\mu}$ indicates the partial trace on the qubit over which $V_\mu$ acts. In Fig.~\ref{fig:proof-1}(b) we show the tensor  representation of this integral. 

Next, we want to integrate the next gate $V_{\mu'}$ such that $U_B^\mu=U_B^{\mu'} V_{\mu'}$ as shown in Fig.~\ref{fig:proof-1}(c). Here we define $\widetilde{\rho}_{\mu'}=U_B^{\mu'}\rho(U_B^{\mu'})^\dagger$ Using Eq.~\eqref{eq:Haar2moment} we find that 
\begin{align}
    &\int_{\mathbb{U}(2)}dV_\mu  \Tr[\Tr_{\mu}[V_{\mu'}\widetilde{\rho}_{\mu'}(V_{\mu'})^\dagger]\Tr_{\mu}[\widetilde{O}_\mu]]\\
    &\quad\quad\quad \quad=\frac{1}{4}\Tr[\Tr_{\mu,\mu'}[\widetilde{\rho}_{\mu'}]\Tr_{\mu,\mu'}[\widetilde{O}_\mu]]\,,
\end{align}
where $\Tr_{\mu}$ indicates the partial trace on the qubit over which $V_\mu$ and $V_{\mu'}$ act.  In Fig.~\ref{fig:proof-1}(c) we show the tensor  representation of this integral.  Iteratively repeating this procedure and integrating over each two-qubit gate in the circuit leads to 
\begin{equation}
    \mathbb{E}_{\vec{\theta}_0}[f(\hat{\vec{\theta}}_{i})]\propto \Tr[O]\Tr[ \st{\psi}]=0 \,,\quad \forall i\,.
\end{equation}
Here we have used the fact that $O$ is a traceless operator. 

\end{proof}

\begin{lemma}\label{lem:p2}
Let $U(\vec{\theta})$ be a shallow HEA with depth $D\in\mathcal{O}(\log(n))$, and $O=\sum_{i}c_i\bigotimes_{\alpha}O_{\alpha}^{(i)}$ with $O_{\alpha}^{(i)}$ being traceless operators having support on at most two neighboring qubits. Then, 
\begin{equation}
    \mathbb{E}_{\vec{\theta}_0}[\Delta f_{i+1,i}\Delta f_{j+1,j}]=0\,,\quad \forall i\neq j\,.
\end{equation}
\end{lemma}

\begin{proof}
First, let us recall that 
\begin{align}
    \Delta  f_{i+1,i}&=f(\vec{\theta}_{i+1})-f(\vec{\theta}_{i})\nonumber\\
    &=2\sin\left(\frac{l_i}{2}\right)\partial_i f(\hat{\vec{\theta}}_i)\,.
\end{align}
where we have defined   $\hat{\vec{\theta}}_{i}=\vec{\theta}_{i-1}+ \hat{\vec{e}}_{i} \frac{l_{i}}{2}$ and used the fact we can apply the  parameter shift-rule for computing gradients~\cite{mitarai2018quantum,schuld2019evaluating}. Then, we have that 
\begin{equation}
    \mathbb{E}[\Delta f_{i+1,i}\Delta f_{j+1,j}]=\kappa_{i,j}\mathbb{E}[\partial_i f(\hat{\vec{\theta}}_{i})\partial_j f(\hat{\vec{\theta}}_j)]\label{eq:exp-grad}\,,
\end{equation}
where $\kappa_{i,j}=4\sin\left(\frac{l_i}{2}\right)\sin\left(\frac{l_j}{2}\right)$.

Assuming without loss of generality that the gate $V_i$ containing parameter $\theta_i$ acts before the gate $V_j$ containing parameter $\theta_j$, and explicitly computing the partial derivatives leads to 
\begin{equation}\label{eq:fiX}
    \partial_if(\vec{\theta}_i)=i\Tr[V_B^i \widetilde{\rho}_i (V_B^i)^\dagger X_i ]\,.
\end{equation}
Here, if $V_i$ is the two-qubit gate containing parameter $\theta_i$ then we define $U(\vec{\theta})=U_A^i V_i U_B^i$  (see Fig.~\ref{fig:proof2}(a)), and we denote  $\partial_i V_i=-i V_A^i H_i V_B^i$, such that  $\partial_i U(\vec{\theta})=-i U_A^i V_A^i H_i V_B^i U_B^i $, where $U_A^i$ and $U_B^i$ contain all other two-qubit gates in the HEA except for $V_i$. In what follows we will assume that both $V_A^i$ and $V_B^i$ also form $2$-designs if $V_i$ forms a $2$-design.  Then, we defined
\begin{align}
\widetilde{\rho}_i&=U_B^i \st{\psi} (U_B^i)^\dagger\\
X_i&=[H_i,(V_A^i)^\dagger (U_{A}^{i})^\dagger O U_A^i V_A^i]\,.
\end{align}
Then, one can similarly define 
\begin{equation}\label{eq:fjX}
    \partial_jf(\vec{\theta}_j)=i\Tr[V_A^i V_B^i \widetilde{\rho}_i (V_B^i)^\dagger (V_A^i)^\dagger X_j ]\,,
\end{equation}
where 
\begin{align}
X_j&=  U_{ij}^{\dag}(V_{B}^{j})^{\dag} [H_j, (V_{A}^{j})^{\dag}(U_A^j)^\dagger O U_A^j(V_{A}^{j}) ] V_{B}^{j}U_{ij} \,.
\end{align}
Here we have defined $U_A^i=U_A^j V_j U_{ij}$ (see Fig.~\ref{fig:proof2}(b)). That is $U_{ij}$ contains all two-qubit gates in the HEA after $V_i$ but before $V_j$. Note here that the parametrized gates in $ \partial_if(\hat{\vec{\theta}}_i)$ and $ \partial_jf(\hat{\vec{\theta}}_j)$ are not necessarily evaluated at the same set of parameters. While this is also true for $V_i$ one can also use the right-invariance of the Haar measure to assume that $V_A^i$ and $V_B^i$ have the same set of parameters in $f(\hat{\vec{\theta}}_i)$ and $f(\hat{\vec{\theta}}_j)$ up to some  unitary that can be absorbed into $U_{ij}$ and $U_B^i$, respectively.

\begin{figure}
    \centering
    \includegraphics[width=.9\linewidth]{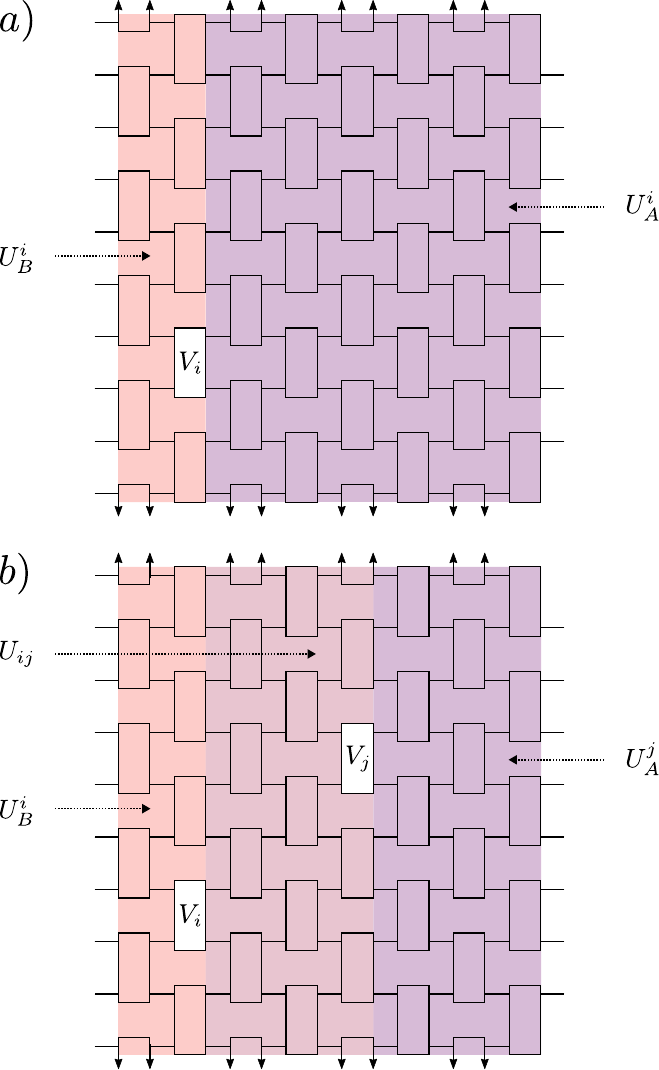}
    \caption{For the proofs, it is useful to divide the HEA into different parts: acting before or after $V_i$ (a), or before, after and in-between $V_i$ and $V_j$ (b).  }
    \label{fig:proof2}
\end{figure}

From the previous, we want to compute $\mathbb{E}[\partial_i f(\hat{\vec{\theta}}_{i})\partial_j f(\hat{\vec{\theta}}_j)]$, which results in computing the expectation value $
    \mathbb{E}[\Tr[V_B^i \widetilde{\rho}_i (V_B^i)^\dagger X_i ]\Tr[V_A^i V_B^i \widetilde{\rho}_i (V_B^i)^\dagger (V_A^i)^\dagger X_j ]]$. 
Using the fact that $V_i$ acts non-trivially only on two qubits, we have the following identity (valid for any operators $A$ and $B$~\cite{cerezo2020cost})
\begin{align}
\Tr\left[(\mathds{1}_{\overline{i}}\otimes V_i)A(\mathds{1}_{\overline{i}}\otimes V_{i}^{\dagger}) B\right]=\sum_{\vec{p},\vec{q}}\Tr\left[V_{i} A_{\vec{q}\vec{p}}V_{i}^{\dagger} B_{\vec{p}\vec{q}}\right]\,,\label{eq:pq}
\end{align}
where the summation runs over all bitstrings $\vec{p},\vec{q}$ of length $(n-2)$ and where
\begin{align}
    A_{\vec{q}\vec{p}}&=\Tr_{\overline{i}}\left[(|\vec{p}\rangle\langle\vec{q}|\otimes\mathds{1}_4)A\right]\\ B_{\vec{p}\vec{q}}&=\Tr_{\overline{i}}\left[(|\vec{q}\rangle\langle\vec{p}|\otimes\mathds{1}_4)B\right]\,.
\end{align}
Here $\mathds{1}_2$ denotes the $4\times 4$ identity and $\Tr_{\overline{i}}$ refers to the trace of all qubits except those over which $V_i$ acts on. We refer the reader to~\cite{cerezo2020cost} for proof of Eq.~\eqref{eq:pq}. Then, an explicit integration over $V_B^i$ leads to 
\begin{equation}
\int dV_B^i \Tr[V_B^i \widetilde{\rho}_i (V_B^i)^\dagger X_i ]\Tr[V_A^i V_B^i \widetilde{\rho}_i (V_B^i)^\dagger (V_A^i)^\dagger X_j ]\nonumber\\
     =\frac{1}{15}\sum_{\substack{\vec{p},\vec{q}\\\vec{p}',\vec{q}'}} \Delta\Psi_{\vec{p}\vec{q}}^{\vec{p}'\vec{q}'}\Tr[\Gamma_{\vec{p}\vec{q}}\Gamma'_{\vec{p}'\vec{q}'}]\,.\label{eq:SM-varWA}
\end{equation}
Here we have defined 
\begin{align}
    \Delta\Psi_{\vec{p}\vec{q}}^{\vec{p}'\vec{q}'} &=  \Tr[\Psi_{\vec{p}\vec{q}}\Psi_{\vec{p}'\vec{q}'}]-\frac{\Tr[\Psi_{\vec{p}\vec{q}}]\Tr[\Psi_{\vec{p}'\vec{q}'}]}{2^{m}}\,,  \label{eq:DeltaPsi}
\end{align}
with 
\begin{align}
    \Psi_{\vec{p}\vec{q}}&=\Tr_{\overline{i}}\left[(|\vec{p}\rangle\langle\vec{q}|\otimes\mathds{1}_4)\widetilde{\rho}_i\right]\,,\\
    \Psi_{\vec{p}'\vec{q}'}&=\Tr_{\overline{i}}\left[(|\vec{p}'\rangle\langle\vec{q}'|\otimes\mathds{1}_4)\widetilde{\rho}_i\right]\,,\\
    \Gamma_{\vec{q}\vec{p}}&=\left[H_i, (V_A^i)^\dagger\Omega_{\vec{q}\vec{p}}  V_A^i\right]\,,\label{eq:gamma-1}\\
    \Gamma'_{\vec{q}'\vec{p}'}&= (V_A^i)^\dagger\Omega'_{\vec{q}'\vec{p}'}  V_A^i\,,\label{eq:gamma-2}
\end{align}
and where 
\begin{align}
    \Omega_{\vec{q}\vec{p}} &=\Tr_{\overline{i}}\left[((|\vec{q}\rangle\langle\vec{p}|\otimes\mathds{1}_4)(U_{A}^{i})^{\dag}OU_{A}^{i}\right]\,,\\
    \Omega'_{\vec{q'}\vec{p'}} &=\Tr_{\overline{i}}\left[((|\vec{q}'\rangle\langle\vec{p}'|\otimes\mathds{1}_4)X_j\right]\,.
\end{align}
Using Eqs.~\eqref{eq:gamma-1} and~\eqref{eq:gamma-2} we  have that 
\begin{align}
    \Tr[\Gamma_{\vec{p}\vec{q}}\Gamma'_{\vec{p}'\vec{q}'}]=& \Tr[H_i (V_A^i)^\dagger\Omega_{\vec{q}\vec{p}} \Omega'_{\vec{q}'\vec{p}'}  V_A^i]- \Tr[\Omega_{\vec{q}\vec{p}}  V_A^i  H_i   (V_A^i)^\dagger\Omega'_{\vec{q}'\vec{p}'}]\,.
\end{align}
Then, integrating over $V_A^i$ leads to
\begin{equation}\label{eq:final-int-zero}
  \int dV_A^i  \Tr[\Gamma_{\vec{p}\vec{q}}\Gamma'_{\vec{p}'\vec{q}'}]=0\,,
\end{equation}
as the integration will lead to terms of the form $\Tr[H_i]$ which are equal to zero since $H_i$ are traceless operators.

Combining Eqs.~\eqref{eq:exp-grad},~\eqref{eq:fiX},~\eqref{eq:fjX} and~\eqref{eq:final-int-zero} shows that
\begin{equation}
    \mathbb{E}[\Delta f_{i+1,i}\Delta f_{j+1,j}]=0\,,\quad \forall i\neq j\,.
\end{equation}
\end{proof}

\begin{lemma}
Let $U(\vec{\theta})$ be a shallow HEA with depth $D\in\mathcal{O}(\log(n))$ where each local two-qubit gates forms a $2$-design on two qubits, and let $O=\sum_{i}c_i\bigotimes_{\alpha}O_{\alpha}^{(i)}$ be the measurement composed of, at most, polynomially many traceless Pauli operators  $O_{\alpha}^{(i)}$ having support on at most two neighboring qubits, and where $\sum_i c_i^2\in\mathcal{O}(\poly(n))$. Then, if the input state follows an area law of entanglement, we have 
\be
\Var[\partial_{\mu}f(\vec{\theta})]\in\Omega\left(\frac{1}{\poly(n)}\right)
\ee
\end{lemma}

\begin{proof}
This lemma leverages the results in Theorem 2 of Ref.~\cite{cerezo2020cost}. Namely, therein it was proved that for a parameter $\theta_\nu$ in  gate $V$ of the HEA
\be
\Var[\partial_{\nu}f(\vec{\theta})]\geq G_{n}(D,\ket{\psi})\,,
\ee
with 
\footnotesize
\begin{align}\label{eq:varMaint22}
    G_{n}(D,\ket{\psi})&=
    \left(\frac{1}{5}\right)^{\scriptscriptstyle 2D}\!\!\!\frac{2 }{225}
    \sum_{ i\in i_{\mathcal{L}}} \sum_{\substack{(k,k')\in k_{\mathcal{L}_{\text{B}}} \\ k' \geq k}} c_i^2  \eta(\psi_{k,k'}) \eta(\tilde{O}_i)\,,
\end{align}
\normalsize
where $i_{\mathcal{L}}$ is the set of $i$ indices whose associated operators $\tilde{O}_i=\bigotimes_{\alpha}O_{\alpha}^{(i)}$ act on qubits in the forward light-cone $\mathcal{L}$ of $V$, and $k_{\mathcal{L}_{\text{B}}}$ is the set of $k$ indices whose associated subsystems $S_k$ are in the backward light-cone $\mathcal{L}_{\text{B}}$ of $V$. Moreover, we recall that we defined  $\eta(M)=\norm{M-\Tr[M]\frac{\mathds{1}}{d_M}}_2$, as the Hilbert-Schmidt distance between $M$ and $\Tr[M]\frac{\mathds{1}}{d_M}$,  where $d_M$ is the dimension of the matrix $M$. In addition, $\psi_{k,k'}$ is the partial trace of the input state $\ket{\psi}$ down to the subsystems $S_k S_{k+1}... S_{k'}$.

Let us note that since $\tilde{O}_i$ is a tensor product of two single-qubit Pauli operators acting on adjacent qubits, then $\eta(\tilde{O}_i)=4$. Then, if the input state $\ket{\psi}$ follows an area law of entanglement, we have that 
\be
I_{\Lambda}(\psi)\in \Omega\left(\frac{1}{\poly(n)}\right)\,.
\ee
In particular, noting that $I_{\Lambda}(\psi)^2\leq 2^{\Lambda-1}\eta_{\psi_\Lambda}$~\cite{coles2019strong} and recalling that in our case $\Lambda$ is the subsystems of qubits $S_k S_{k+1}... S_{k'}$ containing at most a logarithmic number of qubits, one has that 
\begin{equation}
    \eta_{\psi_\Lambda}\in\Omega(1/\poly(n)).
\end{equation}
Finally, since $D\in\mathcal{O}(\log(n))$ we find from Eq.~\eqref{eq:varMaint22} that  \begin{align}
    G_{n}(D,\ket{\psi})\in \Omega\left(\frac{1}{\poly(n)}\right)\,,
\end{align}
which in turn implies 
\be
\Var[\partial_{\nu}f(\vec{\theta})]\in \Omega\left(\frac{1}{\poly(n)}\right)\,.
\ee
\end{proof}
\end{document}